\documentclass[11pt]{article}

\newif\ifnotes

\usepackage[margin=1in]{geometry}
\usepackage{amssymb}
\usepackage{amsmath}
\usepackage{amsthm}
\usepackage{amsfonts}
\usepackage{bm}
\usepackage{color}
\usepackage{comment}
\usepackage[dvipsnames]{xcolor}
\usepackage{float}
\usepackage[T1]{fontenc}
\usepackage{todonotes}
\usepackage{asymptote}
\usepackage{mdframed}
\usepackage[most]{tcolorbox}
\usepackage{framed}
\usepackage{mdframed}
\usepackage{scrextend}
\usepackage{thm-restate}
\usepackage{xcolor}
\usepackage{algorithm}
\usepackage{float}
\usepackage{tikz}
\usepackage{tkz-graph}
\usepackage[hyphenbreaks]{breakurl}

\definecolor{pastelred}{rgb}{1.0, 0.41, 0.38}
\usepackage[commentColor=gray]{algpseudocodex}
\usepackage{algorithm}

\usepackage{etoolbox}
\usepackage{stmaryrd}
\usepackage{thmtools}

\usepackage{graphicx}
\usepackage{dsfont}
\usepackage{mathtools}
\usepackage{bbm}
\usepackage{complexity}
\usepackage{enumitem}
\usetikzlibrary{positioning}
\usepackage{float}
\usetikzlibrary{calc}
\usetikzlibrary{shapes,decorations.markings}
\usetikzlibrary{shapes.geometric}
\usetikzlibrary{positioning, fit, decorations.pathreplacing, arrows.meta }
\usepackage{caption}
\usepackage{array}
\usepackage{subcaption}
\algnewcommand{\LineComment}[1]{\State \textcolor{gray}{// #1 }}

\usepackage[T1]{fontenc}

\usepackage{hyperref}
\hypersetup{
    colorlinks=true,
    linkcolor=blue,
    filecolor=blue,
    citecolor=red,
    urlcolor=black,
}
\usepackage[hyperpageref]{backref}
\usepackage[capitalize]{cleveref}

\Crefname{theorem}{Theorem}{Theorems}
\Crefname{claim}{Claim}{Claims}
\Crefname{lemma}{Lemma}{Lemmas}
\Crefname{proposition}{Proposition}{Propositions}
\Crefname{corollary}{Corollary}{Corollaries}
\Crefname{definition}{Definition}{Definitions}
\crefname{algocf}{alg.}{algs.}
\Crefname{algocf}{Algorithm}{Algorithms}

\renewcommand{\R}{{\mathbb{R}}}

\renewcommand{\epsilon}{\varepsilon}
\newcommand{\eps}{\varepsilon}
\newcommand{\e}{\eps}

\newcommand{\al}{\alpha}

\newcommand{\f}{\frac}

\newcommand{\cd}{\cdot}

\newcommand{\abs}[1]{\lvert #1 \rvert}

\newcommand{\cceil}[1]{\llceil #1 \rrceil}
\newcommand{\cceilx}[1]{\left\lceil\kern-4pt\left\lceil #1 \right\rceil\kern-4pt \right\rceil}

\newcommand{\inv}[1]{{#1}^{-1}}

\newcommand{\todone}[1]{}
\newcommand{\eat}[1]{}

\let \poly \relax

\DeclareMathOperator{\poly}{poly}

\DeclareMathOperator{\rank}{rank}
\DeclareMathOperator{\disc}{disc}
\DeclareMathOperator{\select}{select}

\DeclareMathOperator*{\argmin}{arg\,min}

\newtheorem{thm}{Theorem}[section] 
\newtheorem{lemma}[thm]{Lemma}
\newtheorem{problem}[thm]{Problem}
\newtheorem*{lemma*}{Lemma}

\newtheorem*{lem*}{Lemma}
\newtheorem{prop}[thm]{Proposition}
\newtheorem{conj}[thm]{Conjecture}
\theoremstyle{definition}
\newtheorem{fact}[thm]{Fact}

\newtheorem{definition}[thm]{Definition}
\newtheorem{claim}[thm]{Claim}
\newtheorem*{claim*}{Claim}
\newtheorem{obs}[thm]{Observation}
\newtheorem{cor}[thm]{Corollary}

\newtheorem{remark}[thm]{Remark}
\AtEndEnvironment{corq}{\null\hfill\qedsymbol}

\declaretheoremstyle[%
  spaceabove=-6pt,%
  spacebelow=6pt,%
  headfont=\normalfont\itshape,%
  postheadspace=1em,%
  qed=\qedsymbol%
]{shortspace}

\newlist{thmlist}{enumerate}{1}
\setlist[thmlist]{label=(\alph{thmlisti}),
                  ref=\thethm(\alph{thmlisti}),
                  noitemsep}

\Crefname{listfact}{Fact}{Facts}

\addtotheorempostheadhook[fact]{\crefalias{thmlisti}{listfact}}

\begin{document}

\title{Optimal quantile estimation: beyond the comparison model}
\author{
    Meghal Gupta\thanks{E-mail: \texttt{meghal@berkeley.edu}. This author was supported by an NSF GRFP Fellowship.}\\UC Berkeley \and 
    Mihir Singhal\thanks{E-mail: \texttt{mihirs@berkeley.edu}. This author was supported by an NSF GRFP Fellowship.}\\UC Berkeley \and 
    Hongxun Wu\thanks{E-mail: \texttt{wuhx@berkeley.edu}. This author was supported by Avishay Tal’s Sloan Research Fellowship, NSF CAREER Award CCF-2145474, and Jelani Nelson’s ONR grant N00014-18-1-2562.}\\UC Berkeley
}
\date{\today}

\maketitle
\begin{abstract}
    Estimating quantiles is one of the foundational problems of data sketching. Given $n$ elements $x_1, x_2, \dots, x_n$ from some universe of size $U$ arriving in a data stream, a \textit{quantile sketch} estimates the rank of any element with additive error at most $\e n$. A low-space algorithm solving this task has applications in database systems, network measurement, load balancing, and many other practical scenarios.

Current quantile estimation algorithms described as optimal include the GK sketch (Greenwald and Khanna 2001) using $O(\eps^{-1} \log n)$ words (deterministic) and the KLL sketch (Karnin, Lang, and Liberty 2016) using $O(\eps^{-1} \log\log(1/\delta))$ words (randomized, with failure probability $\delta$). However, both algorithms are only optimal in the comparison-based model, whereas most typical applications involve streams of integers that the sketch can use aside from making comparisons. 

If we go beyond the comparison-based model, the deterministic q-digest sketch (Shrivastava, Buragohain, Agrawal, and Suri 2004) achieves a space complexity of $O(\eps^{-1}\log U)$ words, which is incomparable to the previously-mentioned sketches. It has long been asked whether there is a quantile sketch using $O(\eps^{-1})$ words of space (which is optimal as long as $n \leq \poly(U)$). In this work, we present a \textit{deterministic} algorithm using $O(\eps^{-1})$ words, resolving this line of work.

\end{abstract}

\thispagestyle{empty}
\newpage
\tableofcontents
\pagenumbering{roman}
\newpage
\pagenumbering{arabic}

\sloppy
\section{Introduction}
Estimating basic statistics such as the mean, median, minimum/maximum, and variance of large datasets is a fundamental problem of wide practical interest. Nowadays, the massive amount of data often exceeds the memory capacity of the algorithm. This is captured by \emph{streaming model}: The bounded-memory algorithm makes one pass over the data stream $x_1, x_2, \dots, x_n$ from a universe $[U] = \{1, \dots, U\}$ and, in the end, outputs the statistic of interest. The memory state of the algorithm is therefore a \emph{sketch} of the data set that contains the information about the statistic and allows future insertions. Here, memory consumption is conventionally measured in units of words, where $1$ word equals $\log n + \log U$ bits. 

Most of these simple statistics can be computed exactly with a constant number of words. But the median, or more generally, the $\phi$-quantile, is one exception. In their pioneering paper, Munro and Paterson \cite{munro1980selection} showed that even an algorithm that makes $p$ passes over the data stream still needs $\Omega(n^{1/p})$ space to find the median. Fortunately, for many practical applications, it suffices to find the $\epsilon$-approximate $\phi$-quantile: Instead of outputting the element of rank exactly $\phi n$, the algorithm only has to output an element of rank $(\phi \pm \epsilon) n$. Such algorithms are called \emph{approximate quantile sketches}. They are actually implemented in practice, appearing in systems or libraries such as Spark-SQL~\cite{armbrust2015spark}, the Apache DataSketches project~\cite{apache}, GoogleSQL~\cite{googlesql}, and the popular machine learning library XGBoost~\cite{chen2016xgboost}. 

There are also other queries the sketch could need to answer: For example, online queries asked in the middle of the stream, or rank queries, where the algorithm is asked to estimate the rank of an element up to $\epsilon n$ error. As finding approximate quantiles is equivalent to answering rank queries. To solve all of them, it suffices to solve the following strongest definition. 

\begin{problem}[Quantile sketch]  \label{prob:all-quantile-sketch}
The problem of quantile sketching (or specifically, $\eps$-approximate quantile sketching) is to find a data structure $A$ taking as little space as possible in order to solve the following problem: Given a stream of elements $\pi = x_1, x_2, \dots, x_n \in [U]$, we define the partial stream $\pi_t = x_1, x_2, \dots, x_t$. For element $x \in [U]$, let $\mathrm{rank}_{\pi_t}(x)$ be be the number of elements in $\pi_t$ that are at most $x$. When a query $x$ arrives at time $t$, then $A$ must output an approximate rank $r$, such that $|r - \mathrm{rank}_{\pi_t}(x)| \leq \epsilon t.$
\end{problem}

Two notable quantile sketches include the Greenwald and Khanna (GK) sketch~\cite{greenwald2001space} using $O(\eps^{-1} \log n)$ words (deterministic) and KLL sketch~\cite{karnin2016optimal} using $O(\eps^{-1} \log\log(1/\delta))$ words (randomized, with failure probability $\delta$). Both algorithms follow the comparison-based paradigm, where the sketch cannot see anything about the elements themselves and can only make black-box comparisons between elements it has stored. They are known to be optimal in this paradigm (\cite{cormode2020tight} shows the GK is optimal for deterministic algorithms and \cite{karnin2016optimal} shows that KLL is optimal for randomized algorithms).

However, most typical applications of quantile sketches apply to streams of integers (or elements of some finite universe), rather than just to black-box comparable objects. For example, the elements in the universe could be one of the following: network response times (with a preset timeout), IP addresses, file sizes, or any other data with fixed precision. This may allow for a better quantile sketch than in the comparison-based model. The best previously-known non-comparison-based algorithm is the q-digest sketch introduced in \cite{shrivastava2004medians}, which is a deterministic sketch using $O(\epsilon^{-1} \log U)$ words. Unfortunately, this isn't really better than the GK sketch, as $n$ is typically much less than $\poly(U)$. On the other hand, the only lower bound we know is the trivial lower bound of $\Omega(\epsilon^{-1})$ words in the regime where $n\leq\poly(U)$ (which holds for both deterministic and randomized algorithms). Motivated by this gap, Greenwald and Khanna, in their survey~\cite{greenwald2016quantiles}, asked if the q-digest algorithm is already optimal, and as such, one cannot substantially improve upon comparison-based sketches.

In this work, we resolve this question fully and provide a deterministic quantile sketch that uses the optimal $O(\epsilon^{-1})$ words. This is the first quantile sketch that goes beyond the comparison-based lower bound (in the natural regime of $n\leq\poly(U)$) and is the first direct improvement on the q-digest sketch in the 20 years since it was proposed. 

\begin{thm} \label{thm:main}
There exists a deterministic streaming algorithm for \Cref{prob:all-quantile-sketch} using $O(\epsilon^{-1})$ words (more specifically, $O(\epsilon^{-1} (\log (\epsilon n) + \log(\epsilon U)))$ bits) of space\footnotemark. %
\end{thm}
\footnotetext{Here, technically, when we write $\log(\epsilon n)$ and $\log(\epsilon U)$, it really should be $\max\{\log (\epsilon n), 1\}, \max\{\log (\epsilon U), 1\}$ to avoid the uninteresting corner cases.}

\begin{table}[h]
    \centering
    \renewcommand{\arraystretch}{1.2} 
    \begin{tabular}{>{\centering}m{2.5cm}|>{\centering}m{3.2cm}|c|>{\centering\arraybackslash}m{6cm}}
    \textbf{Algorithm} & \textbf{Type} & \textbf{Space (words)} & \textbf{Space (bits)} \\
    \hline
    GK sketch \cite{greenwald2001space}& deterministic comparison-based  & $O(\epsilon^{-1} \log (\epsilon n))$ & $O(\epsilon^{-1} (\log^2 (\epsilon n) + \log(\epsilon n) \cdot \log U))$\\
    \hline
    q-digest \cite{shrivastava2004medians} & deterministic bounded-universe & $O(\epsilon^{-1} \log U)$ & $O(\epsilon^{-1} (\log^2 U + \log(\epsilon n) \cdot \log U))$\\
    \hline 
    KLL sketch \cite{karnin2016optimal} & randomized comparison-based & $O(\epsilon^{-1} \log \log (1 / \delta))$ & $O(\epsilon^{-1} \log \log (1 / \delta) \cdot (\log \log (1 / \delta) + \log U) + \log n)$\\
    \hline
    Our algorithm (\cref{thm:main}) & deterministic bounded-universe & $O(\epsilon^{-1})$ & $O(\epsilon^{-1} (\log (\epsilon n) + \log (\epsilon U)))$\\
    \end{tabular}
    \caption{The word and bit complexity of quantile sketches.}
    \label{table:complexity}
\end{table}

Our sketch uses less space than not only the deterministic q-digest and GK sketches but also the randomized KLL sketch, when compared in words. Note that randomized algorithms, like KLL sketch, have failure probabilities and retain their theoretical guarantee only against non-adaptive adversaries. The fact that our algorithm is deterministic provides stronger robustness. As these sketches are already implemented in practice, we hope that our algorithm can help improve the performance of these libraries.

\subsection{Discussion and further directions} \label{sec:further-dirs}

\paragraph{Optimality of our algorithm.} As we discussed earlier, the quantile sketch lower bound of $\Omega(\eps^{-1})$ words only holds in the regime where $n\leq\poly(U)$. However, we conjecture that our algorithm is optimal in general for deterministic algorithms. Specifically, there is a simple example showing any sketch for \Cref{prob:all-quantile-sketch} requires at least $\epsilon^{-1} \log(\epsilon U)$ bits (see \Cref{sec:lb}), but we also need to show a lower bound of $\epsilon^{-1} \log (\epsilon n)$ bits. We make the following conjecture about deterministic parallel counting, which would imply our lower bound because any algorithm for \Cref{prob:all-quantile-sketch} can also solve the $k$-parallel counters problem for $k = \Theta(1 / \epsilon)$.

\begin{conj}[Deterministic parallel counters] \label{conj:parallel-count}
We define the $k$-parallel counters problem as following: There are $k$ counters initiated to $0$. Given a stream of increments $i_1, i_2, \dots, i_n \in [k]$ where $i_t$ means to increment the $i_t$-th counter by $1$, the algorithm has to output the final count of each counter up to an additive error of $n / k$.

We conjecture that any deterministic algorithm for this problem requires at least $\Omega(k \log (n / k))$ bits of memory. 
\end{conj}

This conjecture essentially says that to maintain $k$ counters in parallel, one needs to maintain each counter independently. A recent paper by Aden-Ali, Han, Nelson, and Yu~\cite{aden2022amortized} studies this problem for randomized algorithms. The authors of that paper proved that any randomized algorithm with failure probability $\delta$ must use at least $k \cdot \min(\log (n / k), \log \log (1/\delta))$ bits when $\log \log (1/\delta) = \Omega(k)$. Setting $1 / \delta = 2^{2^k}$, this directly implies a $\Omega(\min\{\epsilon^{-1} \log (\epsilon n), \epsilon^{-2}\})$-bit lower bound for any algorithm solving \Cref{prob:all-quantile-sketch}. Thus, our algorithm is also optimal in the regime when $\epsilon^{-1} > \log(\epsilon n)$.

\paragraph{Improvements in the randomized setting.}

Deterministic algorithms are used at the heart of the randomized ones. Many randomized algorithms (including the algorithm by Felber and Ostrovsky \cite{felber2017randomized}, the KLL sketch \cite{karnin2016optimal}, and the mergeable summary of \cite{agarwal2013mergeable}) follow the paradigm of first sampling a number of elements from the stream and then maintaining them with a careful combination of deterministic sketches. 

As long as $n\leq \poly(U)$, our algorithm is optimal even in the randomized setting, but when this condition is not met, it is possible to do better in the randomized setting. If $n$ is known in advance, one can simply sample $\frac{\log 1 / \delta}{\epsilon^2}$ elements and feeds them into our sketch.\footnote{If $n$ is not known is advance, instead of simple sampling, one can replace the use of GK sketch in KLL with our algorithm. As the compactor hierarchy part of KLL stores only $O(1 / \epsilon)$ elements, it results in the same space complexity as the known $n$ case.} It uses a memory of $O(\epsilon^{-1} (\log \log (1/\delta) + \log U) + \log n)$ bits, which strictly improves that of the KLL sketch. We note that, in the most common regime where $\delta > 1/2^{\epsilon n}$, there is a $\Omega(\epsilon^{-1} (\log \log (1/\delta) + \log \epsilon U))$-bit lower bound for streaming quantile sketches.\footnote{This follows from  the $\epsilon^{-1} \log\epsilon U$ lower bound in  \Cref{sec:lb} (which holds for both deterministic and randomized algorithms), and the aforementioned $k \cdot \min(\log(n/k), \log \log (1/\delta))$ %
 lower bound in \cite{aden2022amortized} (setting $k = 1/\epsilon$).}  So our algorithm is also very close to optimal in the randomized setting as well. 

\subsection{Related works}

\paragraph*{More on quantile sketches.} Early works on quantiles sketches include \cite{munro1980selection,alsabti1997one,manku1998approximate}. Among them, the MRL sketch \cite{manku1998approximate} and its randomized variant from \cite{agarwal2013mergeable} lead to the aforementioned KLL sketch. Another variant of the problem is the \emph{biased} quantile sketches (also called \emph{relative error} quantile sketches), meaning that for queries of rank $r$, the algorithm can only have an error of $\epsilon r$ instead of $\epsilon n$. That is, we require that the $0.1\%$ quantiles are extremely accurate, while the $50\%$ quantile can allow much more error. This question was raised in~\cite{gupta2003counting}; since then, people have proposed deterministic~\cite{cormode2006space,zhang2007efficient} and randomized~\cite{cormode2021relative} algorithms for this problem. There are also other variants such as sliding windows~\cite{arasu2004approximate}, weighted streams~\cite{assadi2023generalizing} and relative value error~\cite{masson2019ddsketch}. In practice, there are also the t-digest sketch~\cite{dunning2019computing} and the moment-based sketch~\cite{gan2018moment}, which do not have strict theoretical guarantees. In particular, \cite{cormode2021theory} shows that there exists a data distribution, such that even i.i.d.\ samples from that distribution can cause t-digest to have arbitrarily large error.

\section{Preliminaries} \label{sec:prelim}

\subsection{Definitions for streams}

Define the \textit{rank} of an element $x$ in a stream $\pi$, denoted $\rank_\pi(x)$, to be the total number of elements of $\pi$ that are less than or equal to $x$. We also define a notion of distance between two streams. For two streams $\pi, \pi'$ of equal length, define their \textit{distance} as follows:
\[d(\pi, \pi') = \max_{x \in [1, U]} \abs{\rank_\pi(x) - \rank_{\pi'}(x)}.\] 
We observe that this distance satisfies some basic properties, i.e., the triangle inequality, and subadditivity under concatenation of streams:
\begin{obs}[Triangle inequality] \label{obs:tri-ineq}
For all streams $\pi, \pi', \pi''$ of the same length,
\[d(\pi, \pi') \le d(\pi, \pi'') + d(\pi'', \pi')\]
\end{obs}
\begin{obs} \label{obs:subadd}
For all streams $\pi, \pi'$ of the same length and $\rho, \rho'$ of the same length,
\[d(\pi \circ \rho, \pi' \circ \rho') \le d(\pi, \pi') + d(\rho, \rho'),\]
where $\pi \circ \rho$ denotes concatenation of the streams $\pi$ and $\rho$.
\end{obs}

\subsection{Other notation}

Throughout this paper, we use standard asymptotic notation, including big $O$ and little $o$. For clarity, we sometimes omit floor and ceiling signs where they might technically be required. 

All logarithms in this paper are considered to be in base 2, and we define the \textit{iterated logarithm} $\log^*(m)$ to be the number of times we need to apply a logarithm to the number $m$ to bring its value below 1.

We also define the function $\cceil x$, for any $x \in \R^+$, to be the smallest power of 2 that is at least $x$. In particular, we always have $x \le \cceil x \le 2x$.

\section{Technical overview} \label{sec:overview}

In this section, we explain the main idea of our algorithm.

First, we get a few technical details out of the way. We will assume for now that we know $n$ beforehand. For this overview, we will focus on describing a data structure that uses $O(\eps^{-1}\log(\eps^{-1})\log\log U)$ words of memory. After that, we will briefly describe the modifications that we perform to bring the space complexity down to $O(\e^{-1})$ words. %

\paragraph{The eager q-digest sketch.}
Before explaining our algorithm, it would be instructive to first reivew the q-digest algorithm because our algorithm is based on it. At a high level, this data structure is a tree where every node represents some subset of the stream elements received so far. The node doesn't store each element exactly, but only an interval that contains all of the elements it represents and a count of how many elements it represents. The version we describe slightly differs from the typical treatment, and we call it eager q-digest. The data structure has the following structure and supports the following operations.

\begin{itemize}
    \item \textbf{Structure:} The eager q-digest is a binary tree of depth $\log U$. The nodes in the bottom level of the tree (which we call the \textit{base level}) correspond left-to-right to each element $1, 2, \dots, U$ in the universe. Each non-base level node corresponds to a subinterval of $[1, U]$ consisting of the base level nodes below it. Each node $u$ represents a subset of $W[u]$ elements ($W[u]$ is the weight/count of the node) that have been received so far; that is, when an element is inserted, it increments the counter $W[u]$ at some node. The $W[u]$ elements that $u$ represents must all be within the node's interval.
    
    \item \textbf{Insertion:} We insert elements into the tree top-down as follows: upon receiving an element $x\in [1,U]$, look at the path from the root to $x$ and increment the counter $W[u]$ of the first non-full node $u$. A node is \emph{full} if its weight is already at capacity, which we set to be $\alpha:=\frac{\epsilon n}{\log U}$. Base level nodes are permitted to exceed capacity.
    
    \item \textbf{Rank queries:} We are given an element $x\in [U]$ for which we want to return the rank. To do this, answer with the total weight of everything on the path from the root to the base node $x$ and everything to the left of that path in the tree. All the elements inserted in nodes to the left of this path must have been less than $x$ (since their intervals only contain elements less than $x$) and all the elements inserted to the right must be larger. As such, the error in the rank estimate is only the sum of nodes along the path (not including $x$), which is bounded by the depth of the tree times the weight of each node above $x$, at most $\alpha \log U = \eps n$.

    \item \textbf{Quantile queries:} We are given a rank $r \in [n]$ for which we want to return an element between the rank-$(r-\e n)$ and the rank-$(r+\e n)$ element of the stream. The ability to do this follows from the ability to answer rank queries, since we can simply perform a binary search.\footnote{This is true in a black-box way; see Section~\ref{sec:practice} for details.}
\end{itemize}

Let us look at an example of an eager q-digest. Each node has capacity (maximum weight) $\al=5$ for this example.

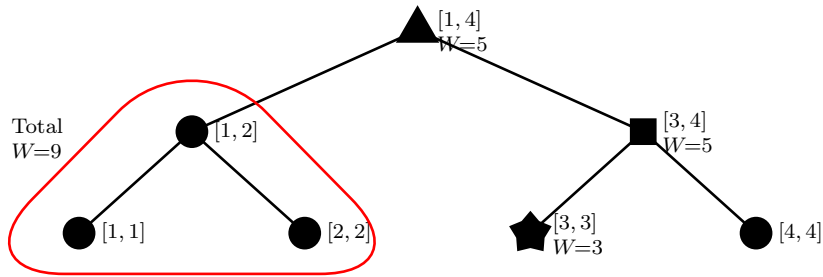
\begin{figure}[ht]
\centering
\begin{tikzpicture}[yscale=0.9,
    level/.style={sibling distance=60mm/#1}, 
    every path/.style={line width=1pt},
    every node/.style={outer sep = -0.5pt},
    filledNode/.style={circle, draw, fill=black, minimum size=4mm}]
    
    \node[filledNode, regular polygon, inner sep=3, regular polygon sides=3,draw,label={[font=\scriptsize]right:\parbox{0.8cm}{$[1,4]$ \\ $W$\hspace{-1pt}$=$5}}] (top) {}
        child {node[filledNode, label={[font=\scriptsize]right:$[1,2]$}] (left1) {}
          child {node[filledNode, label={[font=\scriptsize]right:$[1,1]$}] (left11) {}}
          child {node[filledNode, label={[font=\scriptsize]right:$[2,2]$}] (left12) {}}
        }
        child {node[filledNode, rectangle,draw, scale=0.8, label={[font=\scriptsize]right:\parbox{0.8cm}{$[3,4]$ \\ $W$\hspace{-1pt}$=$5}}] (right1) {}
          child {node[filledNode, star,draw, label={[font=\scriptsize]right:\parbox{0.8cm}{$[3,3]$ \\ $W$\hspace{-1pt}$=$3}}] (right11) {}}
          child {node[filledNode, label={[font=\scriptsize]right:$[4,4]$}] (right12) {}}
        };
    
    \coordinate (A) at ([shift={(0,1.2)}]left1.north);
    \coordinate (B) at ([shift={(-1.6,-0.4)}]left11.south);
    \coordinate (C) at ([shift={(1.6,-0.4)}]left12.south);

    \draw[red,rounded corners=14mm] (A) to (B) to (C) -- cycle;
    \node at ([shift={(-1.9,-1.5)}]A) {\scriptsize{\parbox{1cm}{Total $W$\hspace{-1pt}$=$9}}} ;
\end{tikzpicture}
\caption{An example eager q-digest tree.}
\label{fig:eager-q-digest}
\end{figure}

In this example, triangle represents $5$ elements in the interval $[1,4]$, square represents $5$ more elements in the interval $[3,4]$, and star represents $3$ more elements exactly equal to $3$. If we insert the number $3$ into the example, it would not get inserted into triangle or square because they are full, and so it would be put into star and increment the count by $1$. If we want to then find the rank of the number $3$ (in the pictured tree exactly, before the insertion), we return the sum of the weights on the circled nodes plus the path to $x$, which is $9+5+5+3=22$. This can be off by at most $10$ -- we know the $9$ elements represented by the circled nodes are definitely less than $3$, the ones inserted directly to the star are exactly $3$, the ones to the right are definitely more than $3$. The ones inserted to the triangle and square are the only unknowns. 

\paragraph{Analyzing the space complexity of eager q-digest.}

The space complexity (in bits) of q-digest (and similarly of eager q-digest) is well known to be $O(\eps^{-1}((\log U)^2+\log U\log n))$. Let us understand why, so we can see where we might improve upon this. The space complexity is approximately the product of the following two things:
\begin{enumerate}[label=(\arabic*)]
    \item\label{itm:node-count} The number of non-empty nodes. This is at most $O(\eps^{-1} \log U)$ since the number of full nodes (which is within a constant factor of the number of non-empty nodes) is $n/\alpha=\eps^{-1} \log U$.
    \item\label{itm:node-storage} The amount of space necessary per non-empty node. Naively, we would need to store the location of each nonempty node (the interval it corresponds to) and the weight of the node (the number of stream elements it corresponds to). This would take $\log U+\log n$ space.
\end{enumerate}

As such, in total the space complexity is $O(\eps^{-1}(\log U)^2+\log U\log n)$. In our sketch, we do not reduce~\ref{itm:node-count}, the number of nodes. Therefore, we must reduce the storage in~\ref{itm:node-storage} substantially. This has two parts: efficiently storing the corresponding interval (location in the binary tree) of each node and efficiently storing the count.

It is actually quite simple to store the interval/location of each node: To see this, notice that the non-empty nodes form a connected tree of their own within the large binary tree.  Since the tree is binary, storing the edge from a parent to child in the tree of nonempty nodes takes only $O(1)$ space.
This observation is quite straightforward from the way we formulated q-digest, but the usual implementation of q-digest doesn't push to the top eagerly, and so is unable to directly save this $\log U$ term.

\paragraph{The main challenge: avoiding storing counters.} The second challenge is to avoid storing a counter $W[u]$ at each node. One useful observation about the structure of the tree of non-empty nodes is that all internal nodes are full (at capacity) and only its leaves, which we call \emph{exposed nodes}, need counters. Unfortunately, a constant fraction of the non-empty nodes are exposed nodes, so this doesn't actually save on space. 

Another idea is to store only an approximate count at each node. Unfortunately, we cannot just store an independent approximate count at each node, or even only a counter that estimates when the count surpasses the threshold $\alpha$; this is impossible to do deterministically without using $\log \alpha$ space (which is too large). Even in the randomized setting, approximately counting each node independently does not improve upon KLL.

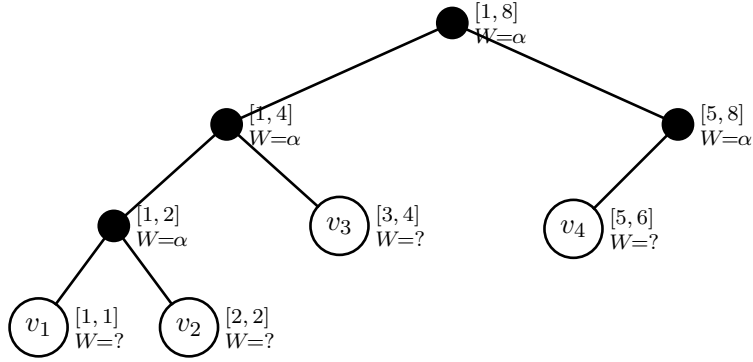
\begin{figure}[ht]
\centering

\begin{tikzpicture}[yscale=0.9,
  level/.style={sibling distance=60mm/#1},
  every node/.style={outer sep = -0.5pt},
  every path/.style={line width=1pt},
  circleNode/.style={circle, draw, fill=black, minimum size=4mm},
  emptyNode/.style={circle, draw, minimum size=4mm}
]
  \node[circleNode, label={[font=\scriptsize]right:\parbox{0.8cm}{$[1,8]$ \\ $W$\hspace{-1pt}$=$\hspace{-0.5pt}$\alpha$}}] (root) {}
    child {node[circleNode, label={[font=\scriptsize]right:\parbox{0.8cm}{$[1,4]$ \\ $W$\hspace{-1pt}$=$\hspace{-0.5pt}$\alpha$}}] {}
      child {node[circleNode, label={[font=\scriptsize]right:\parbox{0.8cm}{$[1,2]$ \\ $W$\hspace{-1pt}$=$\hspace{-0.5pt}$\alpha$}}] {}
        child {node[emptyNode, label={[font=\scriptsize]right:\parbox{0.8cm}{$[1,1]$ \\ $W$\hspace{-1pt}$=$?}}] {$v_1$}}
        child {node[emptyNode, label={[font=\scriptsize]right:\parbox{0.8cm}{$[2,2]$ \\ $W$\hspace{-1pt}$=$?}}] {$v_2$}}
      }
      child {node[emptyNode, label={[font=\scriptsize]right:\parbox{0.8cm}{$[3,4]$ \\ $W$\hspace{-1pt}$=$?}}] {$v_3$}} %
    }
    child {node[circleNode, label={[font=\scriptsize]right:\parbox{0.8cm}{$[5,8]$ \\ $W$\hspace{-1pt}$=$\hspace{-0.5pt}$\alpha$}}] {}
      child {node[emptyNode, below left=of \tikzparentnode, label={[font=\scriptsize]right:\parbox{0.8cm}{$[5,6]$ \\ $W$\hspace{-1pt}$=$?}}] {$v_4$}}
    };
\end{tikzpicture}
\caption{The tree formed by non-empty nodes in eager q-digest. (The filled nodes are the full nodes.)}
\label{fig:exposed-nodes}
\end{figure}

The situation is summarized above in \Cref{fig:exposed-nodes}. At each of the exposed nodes, denoted $v_1,v_2, \ldots, v_\ell$, we want to store some approximate version of counters $W[v_1],W[v_2],\ldots, W[v_\ell]$ that represent how many elements are inserted into that node using significantly less than $\log{n}$ space, ideally $O(1)$ space.

For simplicity, assume that elements are received in ``batches'' of size $\widehat{n}$ (to be determined later), which we can use unlimited space to process. Our only constraint is to minimize storage space between batches.
Let us assume that before the batch, all the counters $W[v_1],W[v_2],\ldots, W[v_\ell]$ are less than $\alpha/2$ and set $\widehat{n}=\alpha/2$ so the set of exposed nodes won't change within the batch. 
At the end of the batch, we need to find suitable approximate values $\widehat{W}[v_1],\widehat{W}[v_2]\ldots, \widehat{W}[v_\ell]$ to increment the counters by, based on the true counts $C[v_1],C[v_2],\ldots, C[v_\ell]$ of the stream elements.

Let us quantify how ``inaccurate'' these approximate counts can be compared to the true counts. The amount of additional error (in rank-space) introduced by answering a rank query for some universe element below a node $v_i$ should be at most $\eps \widehat{n}$ -- we can tolerate this much because it only doubles $\eps$ and we could've chosen $\eps$ to be half as big at the start. The value of this rank query, or the total weight of all the nodes to the left of $v_i$ and the path to $v_i$ changes by $\Big|\Big(\widehat{W}[v_1]+\ldots+\widehat{W}[v_i]\Big)-\Big(C[v_1]+\ldots+ C[v_i]\Big)\Big|$, and so we need to ensure that, for all $i$,
\begin{equation} \label{eq:wi}
    \Big|\Big(\widehat{W}[v_1]+\ldots+\widehat{W}[v_i] \Big)-\Big(C[v_1]+\ldots+ C[v_i]\Big)\Big| < \eps \widehat{n}.
\end{equation}

Here is a simple way to make that happen: Take the $0$-th element, the $(\eps \widehat{n})$-th element, the $(2\eps \widehat{n})$-th element and so on, and increment the counters $W[v_i]$ corresponding to those elements each by $\eps \widehat{n}$. Then, Equation~\ref{eq:wi} is satisfied, and also the counters can be stored in $O(\log(\eps^{-1}))$ bits since they are always multiples of $\eps \widehat{n}=\eps \alpha/2$ and so only have $2\eps^{-1}$ possibilities. 

\paragraph{The main idea: recursive quantile sketch.}
Of course, the glaring issue is how to find (an approximation of) the $0$-th element, the $(\eps \widehat{n})$-th element, the $(2\eps \widehat{n})$-th element and so on, or at least which $v_i$ each one corresponds to, \emph{without} storing the entire batch of $\widehat{n}$ elements. In particular, we have reduced to the following problem: We receive $\widehat{n}$ elements in a stream in the universe $\{v_1, \ldots, v_\ell\}$, and we need to return the approximate $0$-th element, the $(\eps \widehat{n})$-th element, the $(2\eps \widehat{n})$-th element and so on. These are just quantile queries! In particular, we need a quantile sketch on a universe of size $\ell$ receiving $\widehat{n}$ elements. The new universe size $\ell$ is at most the number of exposed nodes of the eager q-digest, which is at most $\eps^{-1}\log U$, and so we have a big saving -- the new quantile sketch is on a logarithmically smaller universe, and so even naively using eager q-digest for the inner sketch will save space.

This solves the problem. The outer quantile sketch requires only $O(\eps^{-1}\log(\eps^{-1})\log U)$ space because it needs $O(\log(\eps^{-1})$ space per node, and the inner sketch requires only $\eps^{-1}\log\log U(\log\log U +\log \widehat{n})$ space because its universe size is $\log U$. Both of these are within $O(\eps^{-1}\log(\eps^{-1})\log\log U)$ words of memory. An illustration of the recursive step is shown in \Cref{fig:recursive-step}, where we build a new eager q-digest whose universe is the exposed nodes of our original eager q-digest. This new eager q-digest will process $\widehat{n}$ elements and ultimately return the $0$-th element, the $\eps \widehat{n}$-th element, $2\eps \widehat{n}$-th element, and so on.

\begin{figure}[ht] \label{fig:recursive-step}
\centering

\begin{tikzpicture}[yscale=0.9,
  level/.style={sibling distance=35mm/#1},
  every path/.style={line width=1pt},
  every node/.style={outer sep = -0.5pt},
  circleNode/.style={circle, draw, fill=black, minimum size=4mm},
  emptyNode/.style={circle, draw, minimum size=4mm}
]
  \begin{scope}[local bounding box=tree1]
    \node[circleNode, label={[font=\scriptsize]right:\parbox{0.8cm}{$[1,8]$ \\ $W$\hspace{-1pt}$=$\hspace{-0.5pt}$\alpha$}}] (root1) {}
      child {node[circleNode, label={[font=\scriptsize]right:\parbox{0.8cm}{$[1,4]$ \\ $W$\hspace{-1pt}$=$\hspace{-0.5pt}$\alpha$}}] {}
        child {node[circleNode, label={[font=\scriptsize]right:\parbox{0.8cm}{$[1,2]$ \\ $W$\hspace{-1pt}$=$\hspace{-0.5pt}$\alpha$}}] {}
          child {node[emptyNode, xshift=-0.3cm, label={[font=\scriptsize]right:\parbox{0.8cm}{$[1,1]$ \\ $W$\hspace{-1pt}$=$0}}] {$v_1$}}
          child {node[emptyNode, xshift=0.3cm, label={[font=\scriptsize]right:\parbox{0.8cm}{$[2,2]$ \\ $W$\hspace{-1pt}$=$0}}] {$v_2$}}
        }
        child {node[emptyNode, label={[font=\scriptsize]right:\parbox{0.8cm}{$[3,4]$ \\ $W$\hspace{-1pt}$=$0}}] {$v_3$}}
      }
      child {node[circleNode, label={[font=\scriptsize]right:\parbox{0.8cm}{$[5,8]$ \\ $W$\hspace{-1pt}$=$\hspace{-0.5pt}$\alpha$}}] {}
      child {node[emptyNode, xshift=-0.5cm, label={[font=\scriptsize]right:\parbox{0.8cm}{$[5,6]$ \\ $W$\hspace{-1pt}$=$0}}] {$v_4$}}
    };
    \node at ($(root1) + (0, -5cm)$) {\bf\vdots};
  \end{scope}

  \draw[->, thick, line width=1.5pt, >=stealth] ([yshift=0mm] tree1.east) -- ++(1.5,0) ;

  \begin{scope}[xshift=7.5cm, local bounding box=tree2]
    \node[emptyNode, label={[font=\scriptsize]right:$[1,6]$}] (root2) {} %
        child {node[emptyNode, label={[font=\scriptsize]right:$[1,2]$}] {} %
            child {node[emptyNode, label={[font=\scriptsize]right:$[1,1]$}] {$v_1$}}
            child {node[emptyNode, label={[font=\scriptsize]right:$[2,2]$}] {$v_2$}}
        }
        child {node[emptyNode, label={[font=\scriptsize]right:$[3,6]$}] {} %
            child {node[emptyNode, label={[font=\scriptsize]right:$[3,4]$}] {$v_3$}}
            child {node[emptyNode, label={[font=\scriptsize]right:$[5,6]$}] {$v_4$}}
    };

    \draw[decorate,decoration={brace,amplitude=5pt},xshift=-0.4cm,thick] (3.4,-3.7) -- (-2.6,-3.7) node[midway, below=0.2cm] {\scriptsize{New base level nodes}};
  \end{scope}
\end{tikzpicture}
\caption{An inner eager q-digest tree whose universe is the exposed nodes of the original tree. (The filled nodes are the full nodes.)}

\end{figure}
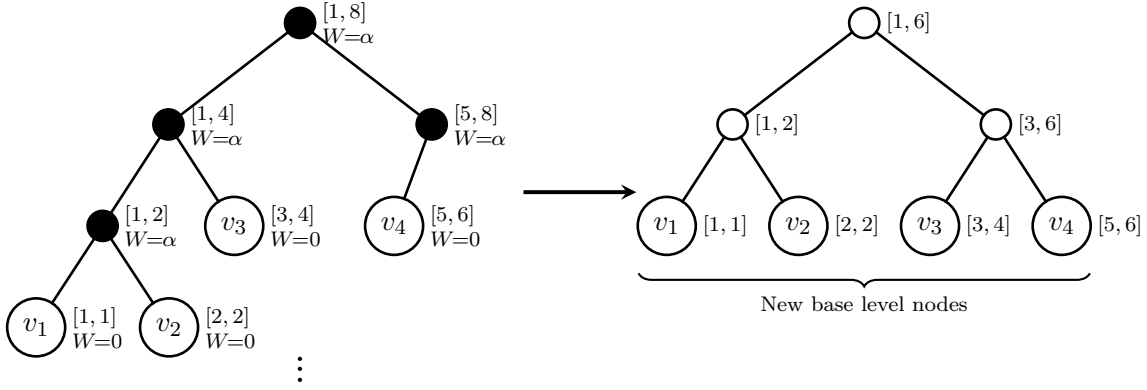

\paragraph{Modifications to get the optimal bounds.}

We can iterate this construction recursively by building a new eager q-digest on the exposed nodes of the second eager q-digest. This process will continue to reduce the universe size nearly logarithmically each time. The number of layers before reaching a constant sized universe is roughly $\log^* U$, and so to get constant error and constant space, we will need to be careful with how we set the error fraction $\eps_i$ for each recursive layer and argue that the total size of the sketches converges.

We also made an assumption that when we started receiving the batch of $\widehat{n}$ elements, all the exposed nodes had weight at most $\alpha/2$. However, the node could have any weight $j\eps\alpha$. To deal with this, we need the lower level q-digest to deal with nodes getting ``overfilled.''

Our final algorithm also manages to get rid of $\log(\eps^{-1})$ factors in the space complexity. This takes a number of additional considerations. One is that the nodes cannot even store counts that require $O(\eps^{-1})$ bits, but truly need to just be either empty or full. To deal with this, we will increase the batch size to $\widehat{n}\alpha$ but now we will need to deal with nodes getting overfilled again. A second issue is that, as described, at the last layer of recursion, the number of nodes would be $\eps^{-1}\log(\eps^{-1})$, which is slightly too large. To deal with this, we will have to use an optimized eager q-digest, which we discuss in Section~\ref{sec:eager}.

\section{Warm up: optimized eager q-digest}  \label{sec:eager}

In this section, we will describe the \textit{optimized eager q-digest} algorithm. This slightly improves the q-digest algorithm of \cite{shrivastava2004medians}. The space complexity of optimized eager q-digest will be fairly similar to that of q-digest (it achieves $O(\eps^{-1} \log \eps n \log \e U)$ bits instead of $O(\eps^{-1} (\log U + \log \eps n) \log U)$ bits). 

Although it does not contain the main idea of this paper, we need it as a building block of our algorithm. Also, we hope that this section can be a warm-up that familiarizes readers with our notation and the basics about q-digest.

Though we have already talked briefly about the eager q-digest in the technical overview, we will start anew in this section by building the algorithm up from the original q-digest, since we make several more modifications than what we described in that section.

\paragraph*{Tree structure of the original q-digest sketch.} In the original q-digest sketch of \cite{shrivastava2004medians}, there is a underlying complete binary tree $T$ of depth $\log U$. We say that those nodes at depth $\log U$ are at the \emph{base-level} of $T$. These nodes correspond (from left to right) to each element $1, 2, \dots, U$ in the universe.  

We label each node in $T$ with a subinterval of $[1, U]$. First, the base-level node corresponding to $i$ is labeled with $[i,i]$. For a node above $u$ the base level, its interval is the union of all its base-level descendants. For every node $u \in T$, it also has a weight $W[u]$ associated to it. Intuitively, one can think of the nodes $u \in T$ with weight $W[u]$ and interval label $[a_u, b_u]$ as a representative of $W[u]$ many elements in the stream that are within $[a_u, b_u]$. 

In the original q-digest all nodes $u$ except the base level nodes can have weight at most $W[u] \leq \alpha$. This is the \emph{capacity} of the node and is usually set to $\alpha = \frac{\epsilon n}{\log U}$. When there is an insertion of stream element $x$, the algorithm finds the base-level node $v$ whose interval equals $[x,x]$ and increases $W[x]$ by $1$.  This is always possible as there is no capacity constraint for base-level nodes. 

Since this tree $T$ has as many as $2U - 1$ nodes, the q-digest algorithm does not store the tree $T$ nor the labels. It only store the set $S$ of \emph{non-empty nodes}, those nodes $v$ with $W[v] > 0$. As there are more and more insertions, the set $S$ grows. Whenever $|S| > \frac{\log U}{\epsilon}$, the q-digest algorithm performs a \emph{compression}. 

One way of performing such compression is to find all nodes $u$ such that $W[u] > 0$ and $W[\mathrm{parent}(u)] < \alpha$, and move one unit of weight from $W[u]$ to $W[\mathrm{parent}(u)]$. After there is no such node $u$, let $F \subseteq S$ be the set of \emph{full nodes} $v$ with $W[v] = \alpha$. We know that $|F| \leq \frac{n}{\alpha} = \frac{\log U}{\epsilon}$. Now for every nonempty node $u \in S$, its parent must be a \emph{full node}. So compression gets the number of nonempty node down to $|S| \leq 3 |F|  = O\left(\frac{\log U}{\epsilon}\right)$. For every $u \in S$, the actual information stored by original q-digest are 1. the position of $u$ in the tree $T$ (which takes $\log U$ bits); 2. weight $W[v]$ (which takes $\log \alpha \approx \log (\epsilon n)$ bits).

Finally, for all these to make sense, we have to be able to answer rank queries. In order to estimate $\rank(x)$, we simply add up the weights $W[u]$ of all nodes $u$ whose intervals contain at least one element less or equal to $x$. This might overcount the number of actual stream elements which are at most $x$; any node whose interval contains both an element which is at most $x$ and greater than $x$ can contribute to the overcounting. These nodes are all (strict) ancestors of the node in the base level corresponding to $x$, so there are at most $\log U$ of them, and their total weight is thus at most $\alpha \cdot \log U$. Thus the answer to the rank query is off by at most $\alpha \cdot \log U \le \eps n$.

Now, having described the original q-digest algorithm, we will describe the modifications we make to it to get the optimized eager q-digest.

\paragraph*{Modification 0: Enforcing capacity constraints on base-level nodes. } In our algorithm, we will need every node, including those at the base level, to satisfy the capacity constraint. But an element $x \in [U]$ could potentially have multiplicity $> \alpha$ in the input stream. 

To handle this, rather than the trees ending at the base level, we let them continue as infinite paths (i.e., unary trees) descending from each node of the base level. Let $u$ be a base-level node that is labeled with interval $[x,x]$. All nodes on the infinite path below $u$ will also just be labeled $[x,x]$. As a sanity check, since we are not storing the tree $T$ anyway, it makes sense to be infinite. 

\paragraph*{Modification 1: Use a forest of $1 / \epsilon$ trees. } To improve the $\log U$ factor to $\log (\epsilon U)$, we have to equally divide the universe into $1 / \epsilon$ intervals and maintain a tree for each one. This gives us a forest of $1 / \epsilon$ trees, while allowing us to set $\alpha$ to $\frac{\epsilon n}{\log(\epsilon U)}$.\footnote{This is because the depth of each tree becomes at most $\log(\epsilon U)$ and the error for answering rank queries is at most the depth multiplied by $\alpha$.}

Roughly speaking, this change corresponds to removing the top $\log(1/\e)$ levels of the q-digest tree while keeping the levels below it. Although only offering a small improvement here, this is actually essential for our final algorithm. It is one of the ingredients that allow us to avoid the extra $O(\epsilon^{-1} \log (\epsilon^{-1}))$ term in the number of words used.

\paragraph*{Modification 2: Move weights up eagerly. } Next we describe how nodes are inserted into the eager q-digest. The original q-digest algorithm moves weight up the tree lazily; that is, it does so when the number of nodes stored exceeds its limit. By contrast, the eager q-digest will do so \textit{eagerly}: upon receiving an element of the stream, it will immediately move it up as much as possible.

More formally, when we receive an element $x$ of the stream, we do not increase the weight of the base-level node with interval $[x,x]$ as we would in a normal q-digest. Instead, we immediately move this weight up. That is, we pick the highest non-full node whose interval contains $x$, and we increment its weight by 1.

\paragraph*{Space Complexity: Full nodes and non-full nodes. }  We now look at the space complexity of optimized eager q-digest. An ordinary q-digest has to store, for every non-empty node, both its location in $T$ and its weight. However, in an optimized eager q-digest, the non-empty nodes are \textit{upward closed}; that is, every parent of a non-empty node is also non-empty. (In fact, every parent of an non-empty node is actually full, since otherwise the weight would have been pushed up to the parent.) Thus, the non-empty nodes form at most $1/\e$ trees which include the roots of their components in $T$. Storing the topology of a binary tree of size $k$ only requires space $k$ (it is enough to use $2$ bits for each node to record whether it has left/right child). Thus the total space required to describe the locations of the non-empty nodes is only $O(|S| + 1/\e)$ bits, where $|S|$ is the total number of non-empty nodes. 

At this point, for all the full nodes, we are already done. Since we know that their weight is exactly $\alpha$, there is nothing more to store. Since $|S| \leq 3|F| \leq \frac{3n}{\alpha} = O\left(\frac{\log(\epsilon U)}{\epsilon}\right)$, we are able to store all the full nodes with only $O(1 / \epsilon)$ words. However, there are still the non-full nodes in $S$. Since we have to store the weight for each of them, this takes $O(|S| \log \alpha) =  O\left(\frac{\log(\epsilon U)}{\epsilon} \cdot \log (\epsilon n)\right)$ space.

This completes the description of eager q-digest. We have saved an $|S| \log U$ term in the space complexity by not having to store the location of each non-empty node, but the $|S| \log \alpha$ term from storing the weights of non-full nodes in $S$ still remains. In the following section, the main idea of our algorithm is to recursively maintain these non-full nodes in $S$ with another recursive layer of our algorithm. When carefully implemented, we are able to ensure that every node in our trees are either full or empty, except at the very last layer of recursion. This removes the extra $|S| \log \alpha$ term. 

\section{Our \texorpdfstring{$O(\epsilon^{-1})$}{O(1/epsilon)}-word algorithm} \label{sec:main}

In this section, we will implement the sketch in \Cref{table:spec}, proving \cref{thm:main}. We assume throughout this section that $\e U$ is at least a sufficiently large absolute constant, since otherwise we can increase $U$ without affecting our asymptotic space complexity.

\begin{table}[h]
    \centering
\begin{tabular}{|m{12.8cm}|}
\hline
\textbf{Our $\epsilon$-approximate quantile sketch} \\
\vspace{0.2cm}
Space complexity: $O(\epsilon^{-1}(\log (\epsilon n) + \log(\epsilon U)))$ bits. \\
\vspace{0.2cm}
Supported operations:
\begin{itemize}
    \item $\Call{Insert}{x}$: Adds an element $x \in [U]$ to the stream. (\Cref{alg:insertion})
    \item $\Call{Rank}{x}$: Returns the rank of $x$ up to $\pm$ $\epsilon t$ error where $t$ is the number of elements in the current stream. (\Cref{alg:query})
\end{itemize}\\
In \Cref{sec:runtime}, we will show that each operation takes $O(\log (1 / \epsilon))$ amortized time under mild assumptions.\\
\hline
\end{tabular}   
    \caption{Our quantile sketch.}
    \label{table:spec}
\end{table}

To start with, we will also assume that we know an upper bound on $n$ (this upper bound will become $n_0$), and that it is sufficiently large (that is, $n_0 \ge n^*$, where $n^*$ is a function of $U, \e$). Furthermore, we will initially allow rank queries to have error up to $\e n_0$. We will maintain these assumptions until \cref{sec:unknown-n}, where we will then describe how to dispense with these assumptions.

We now outline how this section will proceed. In \cref{sec:structure,sec:insertion}, we describe the data structure, and how to handle insertions into the data structure, including how to merge layers of the data structure. In \cref{sec:err_merges}, we bound the error introduced into the data structure with each merge. Then, in \cref{sec:queries}, we describe how to perform rank queries and show a bound of $\e n_0$ on the error of a query. We next, in \cref{sec:unknown-n}, describe how to make our data structure work even when $n < n^*$, and also improve our bound on error of a query to $\e t$ (where $t$ is the size of the stream so far). In \cref{sec:params}, we pick the numerical parameters of our data structure such that the claims of the previous section hold. Finally, in \cref{sec:space,sec:runtime}, we analyze the space and time complexity of our algorithm, respectively.

\subsection{Structure of the sketch} \label{sec:structure}

As mentioned before, our sketch will be formed from recursive applications of the eager q-digest. We now define the structure of the recursive layers, which we number $0, 1, \dots, k$. 

\paragraph*{The $0$-th layer.} We start with the top layer (layer 0) and introduce our notation. The top layer has the same structure as an ordinary optimized eager q-digest forest. We call this underlying forest $T_0$. It has universe size $U_0 = U$ and error parameter $\eps_0 = \epsilon / 8$. We would like to emphasize that in optimized eager q-digest, $T_0$ is a forest with $1 / \epsilon_0$ infinite trees where most nodes have weight $0$. We call these nodes \emph{empty}. %

Whether empty or not, each node in this infinite forest is labeled with an interval. The $1 / \epsilon_0$ roots of the trees are labeled with $[1, \epsilon_0 U], [\epsilon_0 U + 1, 2 \epsilon_0 U], \dots, [(1 - \epsilon_0) U + 1, U]$, respectively. Then, if a node is labeled with interval $[a,b]$, its two children are labeled with $[a,(a+b - 1)/2]$ and $[(a + b + 1) / 2, b]$ respectively. (Since we assumed that $\epsilon$ and $U$ are powers of $2$, these are all integers.) As a special case, if $a = b$, the node is going to have only one child, labeled $[a,a]$.

In $T_0$, each node $u$ has a weight $W_0[u]$ that cannot exceed capacity $\alpha_0 = \eps_0 n_0 / \cceil{\log(\eps_0 U_0) + 1}$, where $n_0$ is an upper bound on $n$. We define the set of \emph{full nodes}, $F_0$, as the set of nodes that have a weight of exactly $\alpha_0$. Recall from optimized eager q-digest that we know $F_0$ is a upward-closed set of nodes and is therefore itself a forest of at most $1/\epsilon_0$ trees. (See \Cref{sec:eager} for details). We will enforce the invariant that every node in the tree $T_0$ is either \emph{full} or \emph{empty}. So nodes in $F_0$ are the only nodes in $T_0$ that we actually use and store. As mentioned before, this allows us to store each node with only a constant number of bits. 

Note that if we were to add new full nodes to this structure, the empty children of full nodes in $F_0$, as well as the empty roots of trees, are potentially positions for new nodes. We call these empty nodes the \emph{exposed nodes}. Formally, the \emph{exposed nodes} of $T_0$ is the set of empty nodes that do not have a full parent. For a concrete example, see the forest $T_0$ in \Cref{fig:main-example}. 

\paragraph{Intuition: Batch processing of insertions.} Let us first jump ahead and sketch the purpose of having layer $i$ $(1 \leq i \leq k)$. Imagine if we insert a new element in the stream. Then, an execution of the eager q-digest algorithm will increase the weight of one exposed node in $V_0$ to $1 \ll \alpha_0$. However, our algorithm cannot do the same, because it would break our invariant of having only full nodes in $T_0$. Instead, we maintain the exposed nodes $V_0$  with our recursive structure (layers $\geq 1$) and insert the new element into layers $\geq 1$. These recursive layers act like a ``buffer''; once they accumulate $n_1$ elements, we clear them and compress those elements into new full nodes in~$T_0$. 

In general, for layer $i$ ($1 \leq i < k$), we group $n_{i + 1}$ insertions in a \emph{batch} and insert them to layer $\geq i + 1$. After each batch, we compress the elements in layer $\geq i + 1$ into full nodes in layer $i$ and clear layer~$\geq i + 1$. Full details of how we handle insertion will be discussed in \Cref{sec:insertion}.

\paragraph*{The $i$-th layer ($1 \leq i \leq k$).} 

Roughly speaking, the upper part of the layer $i$ structure (which we call $T_i$) resembles an optimized eager $q$-digest forest with whose ``universe size'' is $U_i$, which is an upper bound on $|V_{i-1}|$ (when we pick the values of the parameters, we will prove this upper bound in \cref{claim:U-exposed}). %
At depth $h_i \coloneqq \log(\epsilon_i U_i)$, they have exactly $|V_{i-1}|$ nodes\footnote{Note that in an optimized eager $q$-digest, the base level contains $U_i$ nodes; we just remove the remaining $U_i - |V_{i-1}|$ nodes and their descendants, and also any inner nodes with no descendants remaining.}. Each such node $u$ will correspond to an exposed node $v \in V_{i - 1}$, in order (from left to right). We call this depth the \emph{base level} of $T_i$. This is the upper part of $T_i$.

For the interval labeling of the upper part, as each base level node $u$ corresponds to an empty node $v \in V_{i - 1}$, naturally, $u$ just inherits the interval label of $v$. Strictly above the base level, the interval of each node is the union of the intervals of its base-level descendants.  

Now we start to describe the lower part of $T_i$. Unlike the optimized eager $q$-digest, we will also allow $T_i$ to grow beyond the base level. (We give some intuition for this in \Cref{remark:lower_part}, which readers may skip on the first read.) For each base-level node $u$ that corresponds to $v \in V_{i - 1}$, we copy the empty infinite subtree of $v$ in $T_{i- 1}$, and put it as the subtree of $u$ in $T_i$. This also copies the interval labels on nodes in the subtree. For a concrete example, see the forests $T_1, T_2$ in \Cref{fig:main-example}. 

\begin{remark} \label{remark:correspondence}
    Because we copied the subtrees from $T_{i - 1}$, for any node $u$ below the base level (including the base level itself), there exists a unique node $u' \in T_{i - 1}$ corresponding to it. (We will soon see that $u'$ is in fact an empty node.)
\end{remark}

A node $u$ in $T_i$ has weight $W_i[u]$ and capacity $\alpha_i = \e_i n_i / \cceil{\log(\e_i U_i) + 1}$. We again call a node \emph{full} when it reaches its capacity. $F_i$ is defined to be the set of all full nodes in $T_i$. We maintain the similar invariant as layer~$0$:
\begin{center}
For all $0 \leq i < k$, the forest $T_i$ will contain either \emph{full} or \emph{empty} nodes.\footnote{Note when $i = k$, since there are no further recursive layers, we do not require the invariant for it. Insertions to $T_k$ are simply handled as in a normal optimized eager q-digest. (See \Cref{sec:insertion} for more details.)}
\end{center}
Note that this invariant means that, for layers $i < k$, instead of storing the weight map $W_i$, it suffices to only store $F_i$, since the contents of $W_i$ are determined by $F_i$.

Finally, $V_i$, the set of \emph{exposed nodes} of $T_i$, is defined as the set of of empty nodes which do not have a full parent (for $1 \leq i < k$)\footnote{For $i = k$, we define $V_k$ instead to be the set of non-full nodes without a full parent.}. Note that there may be some exposed nodes above the base level. (This results in a subtlety in the interval labels. See \Cref{remark:interval} for details. Readers may skip it on their first read.)

\begin{remark} \label{remark:lower_part}
Suppose that we do not allow the tree $T_i$ to grow beyond the base level. Then the total weight of it can be at most $2 \alpha_i |V_{i - 1}|$. In other words, layer $\geq i$ will not be able to handle more than that many intersections. But it turns out that we will later want to set $n_i \gg 2 \alpha_i |V_{i-1}|$, so we have to allow $T_i$ to grow beyond the base level. (More specifically, we want to set $n_i$ so that $\epsilon_i n_i \geq \alpha_{i - 1}$, which is essential for \Cref{lemma:round}.)
\end{remark}

\begin{remark} \label{remark:interval}
First, for the upper part of $T_i$, a node labeled with $[a,b]$ may not have children with evenly split interval labels ($[a,(a+b - 1) / 2]$ and $[(a + b + 1) / 2, b]$). This is clear since the labels of nodes above base level are derived bottom-up by taking the union of intervals at their base-level descendants. It is, though, tempting to think that for the lower part of $T_i$ (below the base level), all nodes labeled will have two children with equally split intervals. This, however, is also not always the case. It is possible that a base-level node $u$ corresponds to an exposed node $v \in V_{i - 1}$ that is in the upper part of $T_{i - 1}$. Then when we copy the subtree of $v$, those two children will not have equally split interval labels. For example, this happens in \Cref{fig:main-example}, at the node labeled $[1,6]$ in the tree $T_2$. Its two children split into $[1,4]$ and $[5,6]$, while an even split is $[1,3]$ and $[3,6]$. 
\end{remark}

\renewcommand{\labelitemii}{$\circ$}
\begin{remark}
    In order to avoid interrupting the flow of the paper, we will defer the precise definitions of the parameters $k, n_i, U_i, \e_i$ until \cref{sec:params}. However, so that the reader can have a sense of the scale of each of these parameters, we will give approximate values now that can be used as guidelines. All the parameters except $k$ will be powers of 2, to avoid divisibility issues. We pick the following rough values:
    \begin{itemize}[topsep=0pt]
        \item The number of layers will be $k + 1 \approx \log^*(\e U)$.
        \item The approximation parameters $\e_i$ are all very close to $\e$, and can be thought of as essentially equal to $\e$.
        \item The $U_i$ will satisfy the approximate recursion $\e U_{i+1} \approx \log(\e U_i)$, so by the last level we will have $U_k \approx 1/\e$.
        \item The batch sizes $n_i$ will shrink very slowly (only by polylogarithmic factors in $\e U$), so they can all be thought of as roughly $n$, though decreasing.
        \begin{itemize}
            \item In particular, even the last batch size $n_k$ is almost $n$ in this sense, so one can think of the algorithm as spending most of its time at layer $k$, with a ``universe'' of size $O(1/\e)$.
        \end{itemize}
        \item Similarly, the capacities $\al_i$ are also all approximately $\e n$, though also decreasing in $i$.
    \end{itemize}
\end{remark}

\begin{figure}[H]
\centering
\begin{tikzpicture}[xscale=0.9, yscale=0.75,
level/.style={sibling distance=35mm/#1},
every node/.style={outer sep = -0.5pt},
every path/.style={line width=1pt},
filledCircle/.style={circle, draw, fill=black, minimum size=4mm},
emptyCircle/.style={circle, draw, minimum size=4mm},
filledSquare/.style={rectangle, draw, fill=black, minimum size=4mm},
emptySquare/.style={rectangle, draw, minimum size=4mm},
filledTriangle/.style={regular polygon, regular polygon sides=3, inner sep= 2.5, draw, fill=black, minimum size=4mm},
emptyTriangle/.style={regular polygon, inner sep = 2.5, regular polygon sides=3, draw, minimum size=4mm}
]
\node at (-5, -1.5) {\huge{$T_0:$}};

\begin{scope}
\node[emptySquare, label={[font=\scriptsize]right: $[1,4]$}] {}
    child { node[emptyCircle, label={[font=\scriptsize]right: $[1,2]$}] {}
        child { node[emptyCircle, label={[font=\scriptsize]right: $[1,1]$}] {} 
            child[draw=lightgray] { node[emptyCircle, label={[font=\scriptsize, text=gray]right: $[1,1]$}] {} }
        }
        child { node[emptyCircle, label={[font=\scriptsize]right: $[2,2]$}] {} 
            child[draw=lightgray] { node[emptyCircle, label={[font=\scriptsize, text=gray]right: $[2, 2]$}] {} }
        }
    }
    child { node[emptyCircle, label={[font=\scriptsize]right: $[3,4]$}] {}
        child { node[emptyCircle, label={[font=\scriptsize]right: $[3,3]$}] {} 
            child[draw=lightgray] { node[emptyCircle, label={[font=\scriptsize, text=gray]right: $[3, 3]$}] {} }
        }
        child { node[emptyCircle, label={[font=\scriptsize]right: $[4,4]$}] {} 
            child[draw=lightgray] { node[emptyCircle, label={[font=\scriptsize, text=gray]right: $[4, 4]$}] {} }
        }
    };

    \node at (0.2, -5) {\bf\vdots};
\end{scope}

\begin{scope}[xshift=7.5cm]
    \node[filledCircle, label={[font=\scriptsize]right: $[5,8]$}] {}
        child { node[emptySquare, label={[font=\scriptsize]right: $[5,6]$}] {}
            child { node[emptyCircle, label={[font=\scriptsize]right: $[5,5]$}] {} 
                child[draw=lightgray] { node[emptyCircle, label={[font=\scriptsize, text=gray]right: $[5,5]$}] {} }
            }
        child { node[emptyCircle, label={[font=\scriptsize]right: $[6,6]$}] {} 
            child[draw=lightgray] { node[emptyCircle, label={[font=\scriptsize, text=gray]right: $[6, 6]$}] {} }
        }
    }
    child { node[filledCircle, label={[font=\scriptsize]right: $[7,8]$}] {}
        child { node[filledCircle, label={[font=\scriptsize]right: $[7,7]$}] {} 
            child[draw=lightgray] { node[emptySquare, label={[font=\scriptsize, text=gray]right: $[7,7]$}] {} }
        }
        child { node[emptySquare, label={[font=\scriptsize]right: $[8,8]$}] {} 
            child[draw=lightgray] { node[emptyCircle, label={[font=\scriptsize, text=gray]right: $[8, 8]$}] {} }
        }
    };

    \node at (0, -5) {\bf\vdots};
\end{scope}

\draw[->, thick, line width=1.5pt, >=stealth] (4,-5.3) -- (4,-6.8) ;

\node at (-5, -8.5) {\huge{$T_1:$}};

\begin{scope}[yshift=-7cm]
\node[emptyTriangle, label={[font=\scriptsize]right: $[1,6]$}] {}
    child { node[emptySquare, label={[font=\scriptsize]right: $[1,4]$}] {}
        child[draw=lightgray] { node[emptyCircle, label={[font=\scriptsize, text=gray]right: $[1,2]$}] {} 
            child { node[emptyCircle, xshift=0.1cm, label={[font=\scriptsize, text=gray]below: $[1,1]$}] {} }
            child { node[emptyCircle, xshift=-0.1cm, label={[font=\scriptsize, text=gray]below: $[2,2]$}] {} }
        }
        child[draw=lightgray] { node[emptyCircle, label={[font=\scriptsize, text=gray]right: $[3,4]$}] {} 
            child { node[emptyCircle, xshift=0.1cm, label={[font=\scriptsize, text=gray]below: $[3,3]$}] {} }
            child { node[emptyCircle, xshift=-0.1cm, label={[font=\scriptsize, text=gray]below: $[4,4]$}] {} }
        }
    }
    child { node[emptySquare, label={[font=\scriptsize]right: $[5,6]$}] {}
        child[draw=lightgray] { node[emptyCircle, label={[font=\scriptsize, text=gray]right: $[5,5]$}] {} 
            child { node[emptyCircle, label={[font=\scriptsize, text=gray]right: $[5,5]$}] {} }
        }
        child[draw=lightgray] { node[emptyCircle, label={[font=\scriptsize, text=gray]right: $[6,6]$}] {} 
            child { node[emptyCircle, label={[font=\scriptsize, text=gray]right: $[6,6]$}] {} }
        }
    };

    \node at (0.2, -5) {\bf\vdots};
\end{scope}

\begin{scope}[xshift=7.5cm, yshift=-7cm]
    \node[filledCircle, label={[font=\scriptsize]right: $[7,8]$}] {}
        child { node[emptySquare, inner sep=0pt, text depth=-0.25cm, label={[font=\scriptsize]right: $[7,7]$}] {$\boldsymbol{\triangle}$}
            child[draw=lightgray] { node[emptyCircle, label={[font=\scriptsize, text=gray]right: $[7,7]$}] {} 
                child { node[emptyCircle, label={[font=\scriptsize, text=gray]right: $[7,7]$}] {} }
            }
        }
        child { node[filledSquare, label={[font=\scriptsize]right: $[8,8]$}] {}
            child[draw=lightgray] { node[filledCircle, fill=lightgray, label={[font=\scriptsize, text=gray]right: $[8,8]$}] {} 
                child { node[emptyTriangle, label={[font=\scriptsize, text=gray]right: $[8,8]$}] {} }
            }
        };

    \node at (0, -5) {\bf\vdots};
\end{scope}

\draw[->, thick, line width=1.5pt, >=stealth] (4,-12.3) -- (4,-13.8) ;

\node at (-5, -15.5) {\huge{$T_2:$}};

\begin{scope}[yshift=-14cm]
\node[emptyCircle, label={[font=\scriptsize]right: $[1,7]$}] {}
    child { node[emptyTriangle, label={[font=\scriptsize]right: $[1,6]$}] {}
        child[draw=lightgray] { node[emptyCircle, label={[font=\scriptsize, text=gray]right: $[1,4]$}] {} }
        child[draw=lightgray] { node[emptyCircle, label={[font=\scriptsize, text=gray]right: $[5,6]$}] {} }
    }
    child { node[emptyTriangle, label={[font=\scriptsize]right: $[7,7]$}] {}
        child[draw=lightgray] { node[emptyCircle, label={[font=\scriptsize, text=gray]right: $[7,7]$}] {} }
    };

    \node at (0, -4) {\bf\vdots};
\end{scope}

\begin{scope}[xshift=7.5cm, yshift=-14cm]
    \node[emptyCircle, label={[font=\scriptsize]right: $[8,8]$}] {}
        child { node[emptyTriangle, xshift=-1.5cm, text depth=-0.25cm, label={[font=\scriptsize]right: $[8,8]$}] {} 
            child[draw=lightgray] { node[emptyCircle, label={[font=\scriptsize, text=gray]right: $[8,8]$}] {} }
        };
    \node at (0.2, -4) {\bf\vdots};
\end{scope}
\matrix [draw, below=of current bounding box.south, yshift=0.9cm, xshift=1cm, thin, column sep=1.5cm] {
\node[emptySquare, text depth=-0.25cm, label={[label distance=0.2cm]right: Exposed nodes in $T_0$}]{}; &
\node[emptyTriangle, text depth=-0.25cm, label={[label distance=0.2cm]right: Exposed nodes in $T_1$}]{}; &
\node[filledCircle, text depth=-0.25cm, label={[label distance=0.2cm]right: Full nodes}]{};\\
};
\end{tikzpicture}
\caption{The structure of different layers. Here $\epsilon = 0.5$, so there are $1/ \epsilon = 2$ trees in each layer. The nodes below the base-level of each layer is marked as gray. Note that when we construct $T_i$, we take all the exposed nodes in $T_{i - 1}$ and use them as the base-level nodes to build $1/\e$ trees. Then we copy their subtrees in $T_{i - 1}$ to be their subtrees in $T_i$.
}
\label{fig:main-example}
\end{figure}
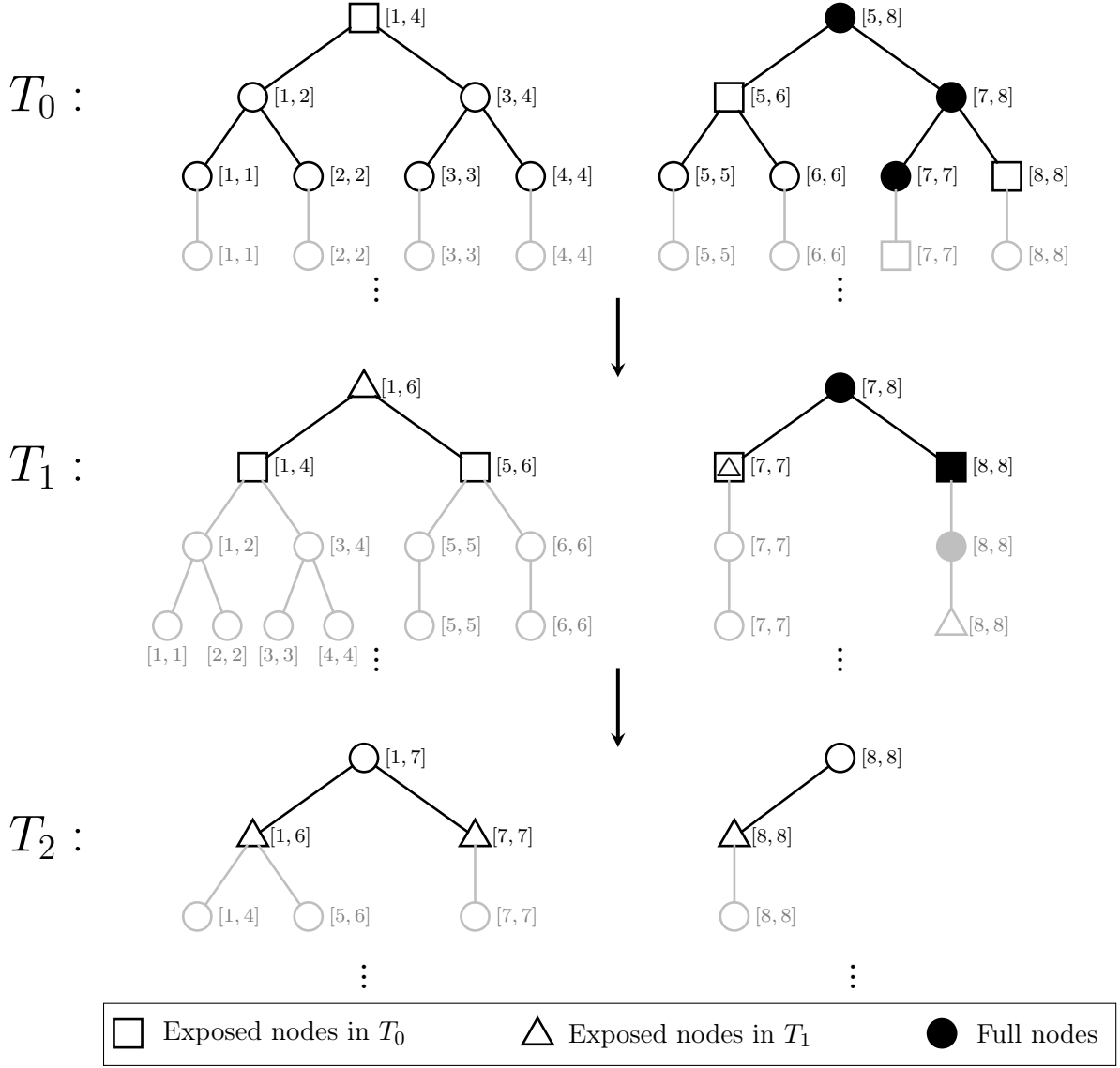

\subsection{Handling insertions} \label{sec:insertion}

In this subsection, we formally explain how we handle insertions. 

\paragraph*{Insertions.} Recall that in \Cref{sec:structure}, we only require our invariant to hold for layers $i \neq k$. For layer $k$, it is maintained by a normal eager q-digest. For any insertion $x$, we first insert it into the layer $k$ as we would in a normal optimized eager q-digest. In other words, we find the exposed node in $T_k$ whose interval contains $x$ and increase its weight, $W_k[v]$, by $1$. This node always exists due to the following observation.

\begin{obs} \label{obs:interval-cover}
For all layers $1 \leq i \leq k$. the intervals of the exposed nodes $V_i$ are always disjoint and cover the entire universe $[1,U]$. 
\end{obs}

Then for $i = k, k - 1, \dots, 1$, we check if the total number of elements inserted so far, denoted by $t$, is a multiple of $n_i$. If so, we need to compress layers $\geq i$ into full nodes in layer $i - 1$. Specifically, we will chose these $n_i$'s so that $n_{i}$ is always a multiple of $n_{i + 1}$ for all $I$ (we prove this in \cref{fact:n-dec}). Therefore if $w_{\text{tot}}$ is a multiple of $n_i$, layers $\geq i + 1$ have already been compressed into full nodes of layer $i$. We will only need to compress layer $i$ into full nodes in layer $i - 1$ and merge them into $T_{i - 1}$. We call this procedure \Call{Merge}{$i$} and will describe it next. The pseudocode for the insertion procedure as a whole is summarized below in \Cref{alg:insertion}. 

\begin{algorithm}[h]
\caption{Inserting an element of the stream}
\label{alg:insertion}
\begin{algorithmic}[1]
\Procedure{Insert}{$x$} 
\LineComment{Insert $x$ into $T_k$.}%
\State $v \gets \text{the highest non-full node in $T_k$ whose interval contains $x$}$ \label{Line:insert-find}%
\State $W_k[v] \gets W_k[v] + 1$ \label{Line:insert-add} \Comment{Recall $W_k[v]$ is the weight of $v$.}
\LineComment{Compress and merge}
\State $t \gets t + 1$ \Comment{$t$ is the total number of stream elements inserted so far.}
\For{$i = k \dots 1$}
\If{$n_k \mid t$}
    \State \Call{Merge}{$i$} \Comment{Compress $T_i$ into full nodes of layer $i - 1$.}
    \State Clear the structure at layer $i$.\LineComment{Note the set $V_{i - 1}$ changes after \Call{Merge}{$i$}. After we clear layer $i$, the structure of $T_i$ implicitly changes according to the new $V_{i - 1}$.}
\EndIf
\EndFor
\If{$t = n_0$}
    \State \Call{Double}{}() \Comment{This handles unknown $n$ (\cref{sec:unknown-n}); the reader may ignore it for now.} \label{line:double}
\EndIf
\EndProcedure
\end{algorithmic}
\end{algorithm}

Next, we explain how \Call{Merge}{$i$} compresses layer $i$ into full nodes in layer $i - 1$. We follow a delicate three-step strategy. On a high level, it is carefully designed so that we incur an error (which is defined formally later in \Cref{sec:err_merges}) of at most $h_i \cdot \alpha_i + \alpha_{i + 1}$ from the compression. (Recall that $h_i \coloneqq \log(\epsilon_i U_i)$ is the depth of the base level in $T_i$.) This is important to our analysis.

\paragraph{Merge Step 1: Move the weight into $T_{i - 1}$.} In the first step, we move all the weight in $T_i$ into empty nodes in $T_{i - 1}$. There are two cases:

\begin{itemize}
    \item For every node $u$ with weight below the base level (including the base level itself) in $T_i$, there is a unique empty node $u'$ in $T_{i - 1}$ corresponding to it. (See \Cref{remark:correspondence}.) We move all the weights for $u$ into that of $u'$. Formally, we just increase weight $W_{k - 1}[u']$ by $W_k[u]$. 
    \item For every node $u$ strictly above the base level of $T_i$, there is no node in $T_{i - 1}$ that directly corresponds to it. Instead, we will take an arbitrary descendant $v \in T_i$ of it at the base level. As $v$ corresponds to an (exposed) empty node $v' \in T_{i - 1}$, we will move the weight of $u$ there. Formally, we increase weight $W_{k - 1}[v']$ by $W_k[u]$.
\end{itemize}

This is summarized in \Cref{alg:push-down}. We defer the error analysis of this step to later in this section. Before we proceed, let us state a simple property about this step. 

\begin{obs} \label{obs:not_exceeding_cap}
We will choose the parameters so that $\alpha_i \cdot h_i \leq \alpha_{i - 1}$ (this will be shown in \cref{fact:al-rec}). (Recall that $h_i \coloneqq \log(\epsilon_i U_i)$ is the depth of the base level in $T_i$.) Thus, after this step, all nodes in $T_{i - 1}$ still have weight at most $\alpha_{i - 1}$.
\end{obs}

Therefore, this step does not exceed the capacity of nodes in $T_{i - 1}$. But it does create a number of non-full nodes: It merges $T_i$ into $T_{i - 1}$ while breaking our invariant of having only full or empty node in $T_{i - 1}$. So the purpose of Step 2 and 3 is exactly to restore this invariant.

\begin{algorithm}[h]
\caption{Moving the weights from layer $i$ to empty nodes in layer $i - 1$}
\label{alg:push-down}
\begin{algorithmic}[1]
\Procedure{Move}{$i$}
\ForAll{$u \in F_i$\footnotemark}
    \If{$u$ is strictly above the base level of $T_i$}
        \State Let $v \in T_i$ be an arbitrary descendant of $u$ at the base level.
        \State Let $v' \in T_{i - 1}$ be the node corresponding to $v$ (by \Cref{remark:correspondence}).\label{line:v_prime}
        \State $W_{i - 1}[v'] \gets W_{i- 1}[v'] + W_i[u]$ 
    \Else        
        \State Let $u' \in T_{i - 1}$ be the node corresponding to $u$ (by \Cref{remark:correspondence}).
        \State $W_{i - 1}[u'] \gets W_{i- 1}[u'] + W_i[u]$ \label{line:u_prime}
    \EndIf
\EndFor
\EndProcedure
\end{algorithmic}
\end{algorithm}
\footnotetext{When $i=k$, this will actually be all $u$ such that $W_k[u]$ is nonzero, rather than just all full nodes.}

\paragraph*{Merge Step 2: Compressing into full nodes.} Naturally, given the non-full nodes in $T_{i - 1}$, we want to first perform a compression step similar to q-digest: Whenever a node $v \in T_{i - 1}$ has a parent that is not full, we move weight from $v$ to $\mathrm{parent}(v)$. 
\begin{algorithm}[h]
\caption{Compressing weights $W_{i-1}$ into full nodes}
\label{alg:push-up}
\begin{algorithmic}[1]
\Procedure{Compress}{$i - 1$}
\While{there exists $v \in T_{i - 1}$, $W_{i - 1}[v] > 0$ and $W_{i - 1}[\mathrm{parent}(v)] < \alpha_{i - 1}$}
    \State $W_{i - 1}[v] \gets W_{i - 1}[v] - 1$
    \State $W_{i - 1}[\mathrm{parent}(v)] \gets W_{i - 1}[\mathrm{parent}(v)] + 1$ \label{line:push-up-end} \Comment{Moving the weights.\footnotemark} 
\EndWhile
\EndProcedure
\end{algorithmic}
\end{algorithm}
\footnotetext{We are keeping the algorithm description simple by moving weights one unit at a time. In an actual implementation, one should of course move the maximum amount possible at each time.}

Let $F_{i - 1}$ be the set of full nodes after this step. We call the nodes that are neither full nor empty \emph{partial nodes}. All the partial nodes are now either non-full children of full nodes in $F_{i - 1}$ or an partially-full root. Importantly, we have the following observation.

\begin{obs} \label{obs:leftover-disjoint}
After this step, the interval labels of the partial nodes are all disjoint.
\end{obs}

This is because no partial node can be an ancestor of another. These partial nodes are the leftovers that we will round up in Step 3. 

\paragraph*{Merge Step 3: Round up the leftovers.} As the interval labels of these leftover partial nodes are disjoint by \Cref{obs:leftover-disjoint}, we can sort these nodes by their interval. Then, roughly speaking, we are going to take the (offline) quantile sketch of these nodes as the result for rounding.

More formally, suppose there are $\ell$ partial nodes. After sorting, these nodes are $v_1, v_2, \dots, v_\ell$. Suppose each partial node $v_j$ is labeled $[a_j, b_j]$. We will have $a_1 \leq b_1 < a_2 \leq b_2  < \cdots < a_{\ell} \leq b_{\ell}$. Let $r = \frac{1}{\alpha_{i - 1}}\sum_{j = 1}^\ell W_{i - 1}[v_j]$ be the number of full nodes that we are expected to round up to.\footnote{This is always an integer, because $\sum_{j = 1}^\ell W_{i - 1}[v_j]$ is equal to $n_i$ minus the total weight in the full nodes formed in Step 1 and 2, and we always choose $n_i$ to be a multiple of $\alpha_{i - 1}$.} For every $m \in [r]$, we find the first $q_m \in [\ell]$ such that $\sum_{j = 1}^{q_m} W_{i - 1}[v_j] \geq m \cdot \alpha_{i - 1}$.  These $v_{q_1}, v_{q_2} \dots, v_{q_r}$ are the ``quantiles'' of these sorted partial nodes.

Then we set the weight of all $v_{q_m}$'s (for all $m \in [r]$) to $\alpha_{i - 1}$ and the weight of all other $v_j$'s to zero. Note these $v_{q_m}$'s must be disjoint since by \Cref{obs:not_exceeding_cap}, any node has weight at most $\alpha_{i - 1}$. This rounds up the partial nodes into $r$ many full nodes and finishes this step. An implementation of this procedure is given below in \Cref{alg:round}.

\begin{algorithm}[h] 
\caption{Rounding the leftovers}
\label{alg:round}
\begin{algorithmic}[1]
\Procedure{Round}{$i - 1$}
\State $c \gets 0$ \Comment{$c$ is the cumulative total weight} \label{line:round-start}
\State $m \gets 1$
\ForAll{partial node $v$ in left-to-right order}
    \State $c \gets c + W_{i - 1}[v]$
    \If{$c \ge m \cdot \al_{i - 1}$}
        \State $W_{i - 1}[v] \gets \alpha_{i - 1}$
        \State $m \gets m + 1$ \label{line:round-end}
    \Else
        \State  $W_{i - 1}[v] \gets 0$
    \EndIf
\EndFor
\EndProcedure
\end{algorithmic}
\end{algorithm}

\paragraph{Conclusion.} Finally, our merging operation is implemented by performing these three steps sequentially.

\begin{algorithm}[h]
\caption{Merging layer $i$ into layer $i-1$}
\label{alg:merge}
\begin{algorithmic}[1]
\Procedure{Merge}{$i$}
\State \Call{Move}{$i$}
\State \Call{Compress}{$i-1$}
\State \Call{Round}{$i-1$} 
\LineComment{During this process; the set of full nodes $F_{i - 1}$ that we store changes as we move weights around. }
\EndProcedure
\end{algorithmic}
\end{algorithm}

\subsection{Error analysis for merges} \label{sec:err_merges}

Before we analyze each step of \Call{Merge}{$i$}, let us first define the error metric. 

\paragraph*{Consistency and Discrepancy.} First, we define the notion of \textit{consistency} between our layer $i$ sketch $T_i$ and a stream of elements $\pi$. Intuitively, this describes what layer $i$ should look like upon receiving stream $\pi$ if the merge had not introduced any error. %

\begin{definition} [Consistency]
We say that a stream $\pi$ is \emph{consistent} with a subset of nodes $S \subseteq T_i$ if and only if there exists a map $f$ that maps $\{1, 2, \dots, |\pi|\}$ to $S$ satisfying the following.
\begin{enumerate}
    \item Each node $u \in S$ is mapped to exactly $W_i[u]$ times.
    \item For every $1 \leq j \leq |\pi|$, the interval label of node $f(j)$ contains $\pi_j$.
\end{enumerate}
\end{definition}

Then we define the discrepancy between $T_i$ and the stream $\pi$. This quantifies the amount of additional error we have. 

\begin{definition}[Discrepancy] \label{def:disc}
We define the discrepancy between a stream $\pi$ and a subset of nodes $S \subseteq T_i$ as 
$$\mathrm{disc}(\pi, S) \coloneqq \min_{\pi' \text{ consistent with } S} d(\pi, \pi').$$
Here, as defined in \Cref{sec:prelim}, the distance between two streams is
\[d(\pi, \pi') = \max_{x \in [1, U]} \abs{\rank_\pi(x) - \rank_{\pi'}(x)}.\]
\end{definition}

\paragraph*{Analysis of Step 1.}  Now, we show that Step 1 increases the discrepancy by at most $\epsilon_i n_i$.

\begin{lemma}[Step 1] \label{lemma:push-down}
Let $T_i$ be the layer-$i$ sketch before \Cref{alg:push-down} (Step 1). Also, let $S$ be the set of originally empty nodes in $T_{i - 1}$ whose weight increases during \Cref{alg:push-down}. 

For any stream $\pi$, we have 
$$\mathrm{disc}(\pi, S) \leq \disc(\pi, T_i) + \epsilon_i n_i.$$
\end{lemma}
\begin{proof}
Let $\pi^* \coloneqq \argmin_{\pi^* \text{ consistent with } T_i} d(\pi, \pi^*)$ and $f^*$ be the consistent mapping from $\pi_*$ to $T_i$. We will construct a stream $\pi'$ and a mapping $f'$ such that $\pi'$ is consistent with $S$ with mapping $f'$ and $d(\pi^*, \pi') \leq \epsilon_i n_i$. This finishes the proof because the distance we define satisfies the triangle inequality $d(\pi, \pi') \leq d(\pi, \pi^*) + \epsilon_i n_i$. 

For any element $\pi^*_j$ ($1 \leq j \leq |\pi^*|$) there are two cases:
\begin{enumerate}
    \item If $f^*(j) = u$ for a node $u$ below the base level of $T_i$, let $u' \in T_{i - 1}$ be the corresponding node (as in \Cref{line:u_prime}, \Cref{alg:push-down}). We let $\pi'_j = \pi^*_j$ and set $f'(j) = u'$.
    \item If $f^*(j) = u$ is a node $u$ strictly above the base level of $T_i$, let $v \in T_i$ be its descendant at the base level and $v' \in T_{i - 1}$ be the corresponding exposed node (as in \Cref{line:v_prime}, \Cref{alg:push-down}). We select an arbitrary element $y$ in the interval of $v$ (which is equal to that of $v'$), and let $\pi'_j = y$. Then we set $f'(y) = v'$.
\end{enumerate}
From this construction, it is clear that $\pi'$ is consistent with $S$ under $f'$. To upper bound $d(\pi^*, \pi)$, consider any query $x \in [1, U]$, the difference of the rank of $x$ in $\pi$ and in $\pi'$ is bounded by the number of $j$'s such that $x$ lies strictly between $\pi^*_j$ and $\pi'_j$. 

As $\pi^*_j \neq \pi'_j$, this can only happen in Case 2. Moreover, as $\pi^*_j$ was initially in the interval of $u$, and $\pi'_j$ is in the interval of $v$ (wich is contained by that of $u$), we know that $x$ must also be in the interval of $u$. Since there are at most $h_i$ such nodes $u$ strictly above the base level of $T_i$, and each is mapped to $\alpha_i$ times, we have at most $h_i \alpha_i$ many such $j$'s. We will choose the parameters in \Cref{sec:params} so that $h_i \alpha_i \leq \epsilon_i n_i$ (this will follow from \eqref{eq:def-al}). This proves $d(\pi^*, \pi) \leq \epsilon_i n_i$.
\end{proof}

Then we need to argue that when $S$ is merged with the original nodes in $T_{i - 1}$, their discrepancies at most add up. This follows from the following observation, which is a consequence of \cref{obs:subadd}:

\begin{obs} \label{obs:disc-merge}
    For two disjoint sets of nodes $S, T$ and any two streams $\pi_1$ and $\pi_2$, we have 
    $$\disc(\pi_1 \circ \pi_2, S \cup T) \leq \disc(\pi_1, S) + \disc(\pi_2, T),$$
    where $\circ$ means concatenating two streams. 
\end{obs}

\paragraph*{Analysis of Step 2.} It is not hard to see that Step 2 never increases discrepancy. 

\begin{lemma}[Step 2] \label{lemma:push-up}
 For any stream $\pi$ that is consistent with $T_{i - 1}$, after we perform \Cref{alg:push-up} on $T_{i - 1}$, $\pi$ is still consistent with $T_{i - 1}$. This implies that for any stream $\pi$, $\disc(\pi, T_{i - 1})$ is always nonincreasing after perform \Cref{alg:push-up} on $T_{i - 1}$.
\end{lemma}
\begin{proof}
We prove this for each operation we perform. Whenever we move one unit of weight from $v$ to $\mathrm{parent}(v)$, we pick an arbitrary $1 \leq j \leq |\pi|$ such that $f(j) = v$ and let $f(j) \gets \mathrm{parent}(v)$. Since the interval of $\mathrm{parent}(v)$ contains that of $v$, the consistency map remains valid. 
\end{proof}

\paragraph*{Analysis of Step 3.}  Finally, we show that the rounding in Step 3 only increases the discrepancy by $\al_{i-1} = \epsilon_i n_i$. 

\begin{lemma}[Step 3] \label{lemma:round}
For any stream $\pi$, whenever we perform Step 3 (\Cref{alg:round}) to $T_{i - 1}$ in our algorithm, the discrepancy $\disc(\pi, T_{i - 1})$ increases by at most $\al_{i-1}$ (which is equal to $\epsilon_i n_i$). 
\end{lemma}
\begin{proof}
First, we only perform \Cref{alg:round} after \Cref{alg:push-up}. So, by \Cref{obs:leftover-disjoint}, all the partial nodes have disjoint intervals before \Cref{alg:round}. 

Before the algorithm starts, let $v_1, v_2, \dots, v_\ell$ be the partial nodes of $T_{i - 1}$ in sorted order, and $r = \frac{1}{\alpha_{i - 1}} \sum_{j = 1}^\ell W_{i - 1}[v_j]$. Suppose $[a_1, b_1], [a_2, b_2], \dots, [a_\ell, b_\ell]$ are their disjoint interval labels. Let $\pi^* = \argmin_{\pi^* \text{ consistent with } T_{i - 1}} d(\pi, \pi^*)$ and $f^*$ be corresponding consistency map. For every $m \in [r]$, let $v_{q_m}$ be the first node such that $\sum_{j = 1}^{q_m} W_{i - 1}[v_j] \geq m \cdot \alpha_{i - 1}$. As discussed in \Cref{sec:insertion}, these $v_{q_1}, v_{q_2}, \dots, v_{q_r}$ are all distinct. 

After the algorithm, all partial nodes become empty, except that $v_{q_1}, v_{q_2}, \dots, v_{q_r}$ become full nodes with weight $\alpha_{i - 1}$. We let $q_0 = 0$. For all $m \in [r]$, we do the following to construct stream $\pi'$ and its consistency map $f'$ (with $T_{i - 1}$ after the algorithm):
\begin{itemize}
    \item For all nodes $v_s$ with $q_{m - 1} < s < q_m$ and all $j \in \{1, 2, \dots, |\pi^*|\}$ such that $f^*(j) = v_s$, we set $\pi'_j \gets a_{q_m}$ and $f'(j) \gets v_{q_m}$.  

    \item For the node $v_{q_{m - 1}}$, we take $\sum_{s = 1}^{q_{(m - 1)}} W_{i - 1}[v_s] - (m - 1) \cdot \alpha_{i - 1}$ many $j$'s such that $f^*(j) = v_{q_{m - 1}}$ and set $\pi'_j \gets a_{q_m}$ and $f'(j) \gets v_{q_m}$. 

    \item For the node $v_{q_m}$, we take $m \cdot \alpha_{i - 1} - \sum_{s = 1}^{(q_m) - 1} W_{i - 1}[v_s]$ many $j$'s such that $f^*(j) = v_{q_m}$ and set $\pi'_j \gets \pi^*_j$ and $f'(j) \gets f^*(j) = v_{q_m}$. 
\end{itemize}
Now we prove that $d(\pi^*, \pi') \leq \alpha_{i - 1}$, which by our choice of parameters in \Cref{sec:params}, will be at most $\epsilon_i n_i$. This will end the proof of this lemma by the triangle inequality $d(\pi, \pi') \leq d(\pi, \pi^*) + d(\pi^*, \pi') \leq d(\pi, \pi^*) + \epsilon_i n_i$.

For any query $x$, its rank in $\pi^*$ and $\pi'$ differs by at most the number of $j$'s such that $x$ is strictly between $\pi^*_j$ and $\pi'_j$. As $\pi^*_j \neq \pi'_j$, this only happens in the first two cases. Suppose $f'(j) = v_{q_m}$. This implies $\pi'_j = a_{q_m}$. Then $f^*(j)$ must be a node $v_s$ with $q_{m - 1} \leq v_s < q_m$, and $\pi^*_j \geq a_{q_{m - 1}}$.

This implies $x$ is in the interval $[a_{q_{m - 1}}, a_{q_m})$. Thus there is a unique $m$ for each query $x$, and by our construction, there can be at most $\alpha_{i - 1}$ many $j$'s that are mapped to $v_{q_m}$ by $f'$. This proves that $d(\pi^*, \pi') \leq \alpha_{i - 1}$.
\end{proof}

\paragraph*{Putting everything together.} We have essentially proved the following lemma. 

\begin{lemma} \label{lemma:merge-error}
Let $\pi$ be the partial stream that arrives at time $[s \cdot n_i + 1, s \cdot n_i]$ for some integer $s$. 
According to \Cref{alg:insertion}, after the $(s \cdot n_i)$-th insertion, we will perform \textup{\Call{Merge}{$k$}}, \textup{\Call{Merge}{$k - 1$}}, \ldots, \textup{\Call{Merge}{$i$}} in order. 

Let $T_i$ be the structure at layer $i$ at the exact point that \textup{\Call{Merge}{$i + 1$}} returns and \textup{\Call{Merge}{$i$}} has not started yet. Then we have
$$\disc(\pi, T_i) \leq 2 \gamma_{i + 1} \cdot n_i,$$
where $$\gamma_{i + 1} = \epsilon_{i + 1} + \epsilon_{i + 2} + \cdots + \epsilon_{k}.$$
\end{lemma}
\begin{proof}
We proceed by induction. In the base case where $i = k$, the layer-$k$ structure $T_k$ is always consistent with the partial stream $\pi$ by construction. Suppose that this holds for $i + 1$. We split the stream $\pi$ into its batches $\pi = \pi^{(1)} \circ \pi^{(2)} \circ \cdots \circ \pi^{(n_i / n_{i + 1})}$ where each $\pi^{(j)}$ has length $n_{i + 1}$. For the ease of notation, we define $\pi^{(1\dots j)} = \pi^{(1)} \circ \pi^{(2)} \circ \cdots \circ \pi^{(j)}$.

By the induction hypothesis, we know that after receiving each $\pi^{(j)}$ but immediately before we perform \Call{Merge}{$i + 1$}, we have $\disc(\pi^{(j)}, T_{i + 1}) \leq 2 \gamma_{i + 2} \cdot n_{i + 1}$. 

Then let us look at the process of \Call{Merge}{$i + 1$} and do another layer of induction. The induction hypothesis is that immediately after receiving $\pi^{(j)}$ and perform \Call{Merge}{$i + 1$}, we have $\disc(\pi^{(1\dots j)}, T_i) \leq 2 (\gamma_{i + 2} + \epsilon_{i + 1}) \cdot j \cdot n_{i + 1}$. When $j = n_i / n_{i + 1}$, this is simply $\disc(\pi, T_i) \leq 2 (\gamma_{i + 2} + \epsilon_{i + 1}) \cdot n_i = 2 \gamma_{i + 1} \cdot n_i$ and proves the outer induction.

In the base case, $T_i$ is empty, and we have $\disc(\emptyset, T_i) = 0$. Suppose for $j - 1$, our induction hypothesis holds. 
\begin{itemize}
    \item It first performs \Call{Merge}{$i + 1$} which, by \Cref{lemma:push-down}, adds a set $S$ of new non-empty nodes to $T_{i - 1}$ with $\disc(\pi^{(j)}, S) \leq \disc(\pi^{(j)}, T_{i + 1}) + \epsilon_{i + 1} \cdot n_{i + 1} \leq (2 \gamma_{i + 2} + \epsilon_{i + 1}) \cdot n_{i + 1}$. Then by \Cref{obs:disc-merge}, after this step, we have $\disc(\pi^{(1\dots j)}, T_i) \leq (2 \gamma_{i + 2} \cdot j + \epsilon_{i + 1} \cdot (2j - 1)) \cdot n_{i + 1}$. 
    \item Then it performs \Call{Compress}{$i$} which, by \Cref{lemma:push-up}, does not increase the discrepancy. 
    \item Finally, it performs \Call{Round}{$i$} which, by \Cref{lemma:round}, increases the discrepancy by at most $\epsilon_{i + 1} \cdot n_{i + 1}$ and results in $\disc(\pi^{(1\dots j)}, T_i) \leq 2(\gamma_{i + 2} + \epsilon_{i + 1}) \cdot j \cdot n_{i + 1}$. 
\end{itemize}
This finishes the inner induction and the proof of this lemma. 
\end{proof}

The inner induction in the proof above actually proves the natural corollary below.

\begin{cor} \label{prop:layer-bound}
Let $\pi$ be the partial stream that arrives at time $[s \cdot n_i + 1, t]$ for some integer $s$ and $t$ such that $n_{i + 1} \mid t$ and $t \leq s \cdot n_i$. After the $t$-th insertion and immediately after \Call{Merge}{$i + 1$} returns.  We have
$$\disc(\pi, T_i) \leq 2 \gamma_{i + 1} \cdot |\pi|$$
where $$\gamma_{i + 1} = \epsilon_{i + 1} + \epsilon_{i + 2} + \cdots + \epsilon_{k}.$$
\end{cor}

\subsection{Answering queries} \label{sec:queries}

To answer a rank query, we simply add up the weights of all the nodes whose interval contains any element that is at most $x$, as shown in \cref{alg:query}.

\begin{algorithm}[h]
\caption{Answering a rank query}
\label{alg:query}
\begin{algorithmic}[1]
\Procedure{Rank}{$x$}
\State $r \gets 0$
\ForAll{$i \in \{0, \dots, k\}$}
    \ForAll{vertices $v \in T_i$}
        \If{the interval of $v$ contains any element less or equal to $x$}
            \State $r \gets r + W_i[v]$
        \EndIf
    \EndFor
\EndFor
\State \Return $r$
\EndProcedure
\end{algorithmic}
\end{algorithm}

First, we bound the total weight of nodes $v$ which could cause over-counting. To this end, we say that a node is \textit{bad} if its interval contains $x$ and, furthermore, its interval is not the length-1 interval containing only $x$. Then, we show the following.
\begin{prop} \label{prop:bad-nodes}
The total weight of all bad nodes, across all layers, is at most $\gamma_0 n_0$.
\end{prop}
\begin{proof}
Let $w_i$ denote the total weight of all bad nodes in $T_i$ (for $i < k$, this is just $\al_i$ times the number of full bad nodes in layer $i$). Moreover, let $c_i$ denote the total \textit{capacity} of all bad nodes in layer $i$, even the empty ones\footnote{For layer $i = k$ specifically, $c_i$ includes also the non-full nodes.} (this is $\al_k$ times the total number of bad nodes in layer $k$).

We will prove the following statement for $0 \le i \le k$ by induction:
\begin{equation} \label{eq:bad-nodes}
w_0 + \dots + w_{i-1} + c_i \le \e_0 n_0 + \dots + \e_i n_i.
\end{equation}
For the base case $i=0$, there are $h_0$ bad nodes in layer $0$ (namely, the strict ancestors of the node in the base level which corresponds to $x$). Therefore we have $c_0 = h_0 \al_0 \le \e_0 n_0$. 

Now, assume that \eqref{eq:bad-nodes} holds for $i-1$ (where $1 \le i \le k$); we will show that it also holds for $i$. Consider the quantity $c_i$, the total capacity of bad nodes in layer $i$. Above the base level of $T_i$, at most one node in each level is bad (since the intervals in a level are disjoint). Thus, the total contribution from these nodes to $c_i$ is at most $h_i \al_i \le \e_i n_i$. 

On the other hand, each bad node in $T_i$ which is at or below the base level corresponds to an \textit{empty} bad node in $T_{i-1}$. Note that the total capacity of empty bad nodes in $T_{i-1}$ is just $c_{i-1} - w_{i-1}$. Moreover, since $\al_i \le \al_{i-1}$ (by \cref{fact:al-rec}), the capacity of each node at or below the base level of $T_i$ is at most the capacity of the corresponding empty bad node in $T_{i-1}$. Thus, the total capacity of bad nodes in $T_i$ which are at or below the base level is at most $c_{i-1} - w_{i-1}$. Therefore, in total, the total capacity of all bad nodes in $T_i$ is at most
\[c_i \le \e_i n_i + c_{i-1} - w_{i-1}.\]
Recall also that by the inductive hypothesis, we have
\begin{equation*}
w_0 + \dots + w_{i-2} + c_{i-1} \le \e_0 n_0 + \dots + \e_{i-1} n_{i-1}.
\end{equation*}
Combining these two inequalities, we recover \eqref{eq:bad-nodes}.

Having proven \eqref{eq:bad-nodes}, it remains to complete the proof of \cref{prop:bad-nodes}. Indeed, setting $i=k$ in \eqref{eq:bad-nodes} and using the fact that $c_i \ge w_i$, the total weight of all bad nodes is at most
\[\e_0 n_0 + \dots + \e_i n_i \le (\e_0 + \dots + \e_i) n_0 = \gamma_0 n_0,\]
as desired.
\end{proof}
    
\begin{prop} \label{prop:overall-appx}
At any time $t$, suppose that $\pi$ is the stream received so far. Then there exists a decomposition $\pi = \pi_0 \circ \pi_1 \circ \dots \circ \pi_k$ such that 
$$\sum_{i = 0}^k  \disc(\pi_i, T_i) \leq 2 \gamma_1 t.$$
\end{prop}
\begin{proof}
Let $\pi_0$ be the first $\lfloor t / n_1 \rfloor \cdot n_1$ elements of $\pi$, $\pi_1$ be the next $\lfloor t / n_2\rfloor \cdot n_2 - |\pi_1|$ elements, $\pi_2$ be the next $\lfloor t / n_3\rfloor \cdot n_3 - |\pi_1 \circ \pi_2|$ elements, and so on. In general, $\pi_i$ is the next $\lfloor t / n_{i + 1}\rfloor \cdot n_{i + 1} - |\pi_1 \circ \pi_2 \circ \cdots \circ \pi_{i - 1}|$ elements in $\pi$ after those in $\pi_{i - 1}$. Specifically, we let $n_{k + 1} = 1$. 

By \cref{prop:layer-bound}, we know that $\disc(\pi_i, T_i) \leq 2 \gamma_{i + 1} \cdot |\pi_i|$. Thus, 
\begin{align*}
\sum_{i = 0}^k  \disc(\pi_i, T_i) &\le 2 \gamma_1 |\pi_0| + 2 \gamma_2 |\pi_1| + \dots + 2 \gamma_k |\pi_{k}| \\
&\le 2 \gamma_1 (|\pi_0| + |\pi_1| + \dots + |\pi_k|) \\
&= 2 \gamma_1 t,
\end{align*}
so we are done.
\end{proof}

These two propositions imply a bound on the error of a rank query:

\begin{prop} \label{prop:error-bound}
Let $\pi$ be the stream received so far at time $t$. Then, the answer to a rank query, as performed by \cref{alg:query}, for any element $x$ differs from $\rank_\pi(x)$ by at most $\gamma_0 n_0 + 2 \gamma_1 t$.
\end{prop}
\begin{proof}
Let $\pi = \pi_0 \circ \pi_1 \circ \cdots \circ \pi_k$ be the decomposition from \cref{prop:overall-appx}. Combine \cref{prop:overall-appx} with the definition of discrepancy (\Cref{def:disc}), we know that there exists a sequence of partial streams $\{\pi'_i\}_{i=0}^k$ such that $\pi'_i$ is consistent with $T_i$ and $\sum_{i=0}^k d(\pi_i, \pi'_i) \leq 2 \gamma_1 t$.

Let $\pi' = \pi'_1 \circ \pi'_2 \circ \cdots \circ \pi'_k$. By the triangle inequality (\cref{obs:tri-ineq}), we know that $d(\pi, \pi') \leq 2 \gamma_1 t$. Since we answered the query by counting the total weight of nodes whose intervals include any element which is at most $x$, the quantity obtained is at least $\rank_{\pi'}(x)$, and may overcount at nodes whose interval also contains an element larger than $x$. However, note that any such node must be bad, so the total amount by which the algorithm overcounts is at most $\gamma_0 n_0$ by \cref{prop:bad-nodes}. Thus, the output of the algorithm differs from $\rank_{\pi'}(x)$ by at most $\gamma_0 n_0$. Furthermore, by \cref{prop:overall-appx} (and the definition of distance of streams), we have $\abs{\rank_\pi(x) - \rank_{\pi'}(x)} \le 2 \gamma_1 t$, so the conclusion follows.
\end{proof}

Now, since $\gamma_0, \gamma_1 \le \e/4$ (by \cref{fact:gamma-bd}), this already means that the error of a rank query is at most $\e n_0$. However, so far we have still assumed that we know $n$ in advance; moreover, we would actually like the error to be at most $\e t$, where $t$ is the total number of elements received so far. In \cref{sec:unknown-n}, we will explain how to rectify this.

\subsection{Removing assumptions about \texorpdfstring{$n$}{n}} \label{sec:unknown-n}
In this section, we will describe how to dispense with the assumption that we know $n$, as well as the assumption that $n \ge n^*$. We will also prove that the error of any query is at most $\e t$.

\paragraph{Unknown $n$.} First, we describe how to maintain the data structure when we don't know $n$ in advance, but still assuming that all queries happen after $t \ge n^*$. At the start of the algorithm, we initialize the data structure with $n_0 = n^*$. Then, whenever $t$, the number of elements so far in the stream, reaches $n_0$, we double $n_0$ (which has the effect of doubling $n_i$ and $\al_i$ for all $i$). Note that when $t=n_0$, only layer $0$ exists, so we only need to describe how to update layer $0$. Every node in layer $0$ is now half-full instead of being full; that is, the weight of every node in $F_0$ is now $\al_0/2$. Then, we just perform the push-up and rounding, as described in \cref{alg:push-up,alg:round}, to layer $0$. The pseudocode of this procedure is given in \cref{alg:double}, and it is called in \cref{line:double} of \cref{alg:insertion}.

By \cref{lemma:round}, this has the effect of changing the stream represented by layer 0 by a distance of at most $\al_0 = \e_0 n_0 / \cceil{\log(\e_0 U_0) + 1} \le \e t / 16$ (since we assumed that $\e U$ is sufficiently large). Then, at any point in the stream, the total amount the represented stream has been changed by these rounding operations is at most $\e n_0 / 16 + \e n_0 / 32 + \dots \le \e n_0 / 8$. Therefore, the bound on distance between $\pi$ and $\pi'$ in \cref{prop:overall-appx} is increased by $\e n_0 / 16$ after adding the doubling step to the algorithm.

Therefore, after this modification to the algorithm, the proof of \cref{prop:error-bound} now gives a bound of $\gamma_0 n_0 + \gamma_1 t + \e n_0 / 16$. Since $n_0 \le 2t$ (since we assumed that $t \ge n^*$), and $\gamma_0 \le \e/4$ and $\gamma_1 \le \e/8$ (by \cref{fact:gamma-bd}), we have
\[\gamma_0 n_0 + \gamma_1 t + \e n_0 / 16 \le \e t.\]
In conclusion, for any $t \ge n^*$, the additive error of any rank query after $t$ elements of the stream is at most $\e t$, as desired. It remains, then, to handle the cases where $t < n^*$.

\begin{algorithm}[h]
\caption{Doubling size of data structure}
\label{alg:double}
\begin{algorithmic}[1]
\Procedure{Double()}{}
\ForAll{$i \in \{0, \dots, k\}$} 
    \State $n_i \gets 2n_i$
    \State $\al_i \gets 2\al_i$ \Comment{The algorithm doesn't actually store $n_i$ or $\al_i$; however, this does affect $F_0$ since the nodes in $W_0$ are now half-full instead of full.}
\EndFor
\State \Call{Compress}{0}
\State \Call{Round}{0}
\EndProcedure
\end{algorithmic}
\end{algorithm}

\paragraph{Dealing with $1/\e \le t < n^*$.} Next, we describe how to modify the algorithm to still be able to answer queries when $1/\e \le t < n^*$. Firstly, we still store the original data structure, since we will need to use it after $t$ exceeds $n^*$. However, in addition, we create a new instantiation of the data structure (with the same parameters), where upon receiving an element of the stream, instead of inserting it once, we insert the same element $\e n^*$ times (by \cref{fact:n-star-bd}, this is an integer). Then, as long as $t \ge 1/\e$, we will have inserted at least $n^*$ elements into this alternate data structure, so by the previous section, it will be able to answer rank queries with relative error at most $\e$, as desired. Of course, the effective value of $t$ will have increased by a factor of $\e n^*$, which will have ramifications for the space complexity. However, we will show in \cref{sec:space} that the space complexity is still what we want it to be.

\paragraph{Dealing with $t < 1/\e$.} Finally, while $t < 1/\e$, we will just store all the elements of the stream so far explicitly (in addition to keeping the data structures of the previous two sections). We will show in \cref{sec:space} that this can actually be done using $O(\e^{-1} \log(\e U))$ space. Obviously, if we store all the elements of the stream, rank queries can be answered exactly. 

\subsection{Choosing the parameters} \label{sec:params}

We will now choose values for the parameters of the algorithm ($k$, $n_i$, $U_i$, and $\e_i$) and verify that they satisfy some necessary properties.

First, note that we may assume that $n$, $U$, and $\e$ are all powers of 2 (by rounding $n$ and $U$ up and $\e$ down to the nearest power of 2, costing at most a constant factor). Indeed, we will ensure that $n_i$, $U_i$, $\e_i$, and $\al_i$ are always powers of 2, in order to stave off divisibility issues.

We then pick the following values. Let $k = \log^* (\eps U)$. As described in \cref{sec:structure}, let $U_0 = U$. Let $n_0$ be an upper bound on $t$, the number of elements so far in the stream. As previously described, we will imagine for now that we know $n$ in advance and that $n_0 = n$. Also, we assume, as we may, that $n_0$ is a power of 2. We then pick $\eps_i$ as follows:
\begin{equation*}
\e_i =
\begin{cases}
\e/8, & i = 0, \\
\e / 2^{k - i + 4}, & i \ge 1.
\end{cases}
\end{equation*}
Also, define
\begin{equation*}
\gamma_i = \e_i + \e_{i+1} + \dots + \e_k.
\end{equation*}
Also, recall from \cref{sec:structure} that for all $i$, we define the capacities $\al_i$ based on $\e_i$, $n_i$, and $U_i$as follows:
\begin{equation} \label{eq:def-al}
\al_i = \f{\e_i n_i}{\cceil{\log(\e_i U_i) + 1}}.
\end{equation}
Now, we define the parameters $n_i$ and $U_i$ for layer $i+1$ recursively (for $i < k$) as follows:
\begin{equation} \label{eq:def-U}
U_{i+1} = \cceilx{\f{1}{\e_i} + \f{n_i}{\al_i}} = \cceilx{\f{1 + \cceil{\log(\e_i U_i) + 1}}{\e_i}} = \f{2\cceil{\log(\e_i U_i) + 1}}{\e_i},
\end{equation}
\begin{equation} \label{eq:def-n}
n_{i+1} = \f{\al_i}{\e_{i+1}} = \f{\e_i n_i}{\e_{i+1} \cceil{\log(\e_i U_i) + 1}}.
\end{equation}
We let $h_i$ be the depth of the base layer in $T_i$:
\begin{equation*}
h_i = \log(\e_i U_i).
\end{equation*}
We will show soon that indeed $h_i$ is always a positive integer.

Now, so far we have treated $n_0$ as fixed, but this assumption will change later in \cref{sec:unknown-n}. In anticipation of this, we will briefly discuss here the effects of changing $n_0$. Treating $\eps, U$ as constants, note that the only parameters that are affected by $n_0$ are the $n_i$ and $\al_i$, which are all constant multiples of $n_0$. We will need all the $n_i$ and $\al_i$ to be integers (or equivalently, at least 1), so to this end, define
\begin{equation} \label{eq:def-n-star}
n^* = \f{n_0}{\alpha_k}.
\end{equation}
Then, $n^*$ is fixed (i.e., it depends only on $\eps, U$ and not on $n_0$). Note that $n^*$ is the value of $n_0$ that causes $\al_k$ to equal 1 (and we will show in \cref{fact:int-caps} that it will also cause the rest of the $n_i, \al_i$ to be integral). 

Now we will check some properties of these parameters which we will need. First, we will show the important property of $U_i$: that it is an upper bound on the number of exposed nodes in the previous layer.

\begin{claim} \label{claim:U-exposed}
For all $0 \le i < k$, we have $U_{i+1} \ge |V_i|$ (recall that $V_i$ is the set of exposed nodes in layer $i$).
\end{claim}
\begin{proof}
The number of full nodes in layer $i$ is at most $n_i / \al_i$ (since full nodes have weight $\al_i$. If there are no full nodes, then we would have $|V_i| = 1/\e_i$, since $V_i$ would just be the set of all the roots of trees in $T_i$. Now, imagine building up the set of full nodes by adding them one at a time (from bottom to top). Each time we add a full node, we remove one exposed node, and add back at most two exposed nodes. Thus, the total number of exposed nodes after this process is at most $1/\e_i + n_i / \al_i$, which is indeed at most $U_{i+1}$ by \eqref{eq:def-U}.
\end{proof}

Now, we will prove various other properties of the parameters which we will need throughout. We state all these properties now, but we will defer their proof to \cref{sec:parameter-proof}, since they mostly just involve manipulation of the definitions of the parameters.

\setlist*[thmlist]{nosep}
\begin{fact} \label{fact:parameter-facts}
The parameters satisfy the following properties:
\begin{thmlist}
\item For all $i$, $n_i$, $U_i$, $\e_i$, and $\al_i$ are powers of 2. \label{fact:pow-2}
\item For all $i$, $\e_i U_i \ge 2$ (and thus, $h_i$ is a positive integer). \label{fact:U-bd} %
\item For all $i < k$, $n_{i+1}$ is a factor of $n_i$. \label{fact:n-dec}
\item For all $i < k$, $\al_{i+1} = \al_i / \cceil{h_{i+1}+1}$. \label{fact:al-rec} %
\item If $n_0 \ge n^*$, then $n_i, \al_i \ge 1$ for all $i$. \label{fact:int-caps}
\item $\gamma_0 \le \e/4$, and $\gamma_i \le \e/8$ for all $i>1$. \label{fact:gamma-bd}
\item $U_k = O(1/\e)$. \label{fact:bottom-layer}
\item $U_1 + U_2 + \dots + U_k = O(\e^{-1} \log(\e U))$. \label{fact:tot-U}
\item $\inv \e \le n^* \le \e^{-1} (\log(\e U))^{1+o(1)}$ (where $o(1)$ refers to a term that approaches 0 as $\e U \to \infty$). \label{fact:n-star-bd}
\item $\al_{k-1} = O(n_0 / n^*)$. \label{fact:al-k-1}
\end{thmlist}
\end{fact}

\subsection{Space complexity} \label{sec:space}

Now we discuss the space complexity of the algorithm. All space complexities in this section will be in bits, not words. There are two primary things to check: the space taken by the sketch itself, and the space required during a merge step after an insertion.

\paragraph{Space of sketch.} The information stored by the algorithm consists only of the full nodes $F_i$ for layers $0 \le i < k$ and the weights $W_k$ for layer $k$. (Note that we don't need to store $T_i$ since it is determined recursively by $T_{i-1}$ and $F_{i-1}$.)

Each $F_i$ is an upward-closed subset of $T_i$. In each of the $1/\e_i$ trees that comprise $T_i$, the portion of $F_i$ in that tree (if nonempty) is a connected subgraph including the root. Thus, that portion of $F_i$ is uniquely determined by the topology of the (rooted) tree that it forms (where in a tree topology we ). We can store the topology of an $\ell$-vertex tree using $O(\ell)$ bits (by storing the bracket representation of the tree). The total number of full nodes in $F_i$ is at most $n_i/\al_i$ at any time, so this means that the total space to store $F_i$ is $O(1/\e_i + n_i/\al_i)$, which is just $O(U_{i+1})$ by \eqref{eq:def-U}. Thus, the total space to store all the $F_i$ is $O(U_1 + \dots + U_k)$, which is $O(\e^{-1} \log(\e U))$ by \cref{fact:tot-U}.

Now, it remains to check the space required to store $W_k$. First, the keys of $W_k$ also form an upward-closed subset of $T_i$. This subset consists of full and partial nodes; by the same argument, there are at most $n_k / \al_k = O(1/\e)$ full nodes. Every partial node is either a root (of which there are $O(1/\e_k) = O(1/\e)$) or a child of a full node, so there are also at most $O(1/\e)$ partial nodes. Therefore, as with the $F_i$, the space required to store the set of all nonempty nodes is at most $O(1/\e_k + 1/\e) = O(1/\e)$.

After the set of nonempty nodes has been stored, we just need to store their weights\footnote{Actually, we only need to store the weights of the leaves of the forest formed by the nonempty nodes, since the r. Since it doesn't make a different to the asymptotic space complexity, we store all the weights for simplicity.} in some order (say pre-order of the trees). The weights are all at most $\al_k = n_0 / n^*$, and there are $O(1/\e)$ of them, so the space required to store all the weights is at most $O(\e^{-1} \log(n_0/n^*))$. Since $n^* \ge 1/\e$ (by \cref{fact:n-star-bd}) and $n_0 \le \max(2n, n^*)$ at all times, we have $O(\inv \e \log(n_0/n^*)) \le O(\inv \e \log(\e n))$.

Putting everything together, the total space complexity of the data structure is at most
\[O(\e^{-1} (\log(\e U) + \log(\e n)),\]
as desired.

\paragraph{Space of sketch while $t$ is small.} Recall that in section \cref{sec:unknown-n}, we made two modifications to the data structure that lasted while $t < n^*$ and $t < 1/\e$. We will show now that (asymptotically) they don't require any extra space.

First, while $t < n^*$, we maintained a second data structure identical to the first, except that we repeated each element $\e n^*$ times. For this data structure, the space analysis that we just performed still holds, except that $n_0$ may now be up to $2 \e n^* t$. The space to store the $F_i$ is unchanged. The space required to store the weights is now at most $O(\e^{-1} \log(n_0/n^*)) \le O(\e^{-1} \log(\e t))$, which is still at most $O(\inv \e \log(\e n))$, as desired.

Finally, for $t < 1/\e$, we stored all the elements of the stream explicitly. Naively, storing these as an ordered list would take $O(\inv e \log U)$ space, but actually, since the set is unordered, we can improve this. Indeed, split the universe $[1, U]$ into $1/\e$ buckets of size $\e U$ (based on the $\log \inv \e$ most significant bits). Then, for each bucket, store an ordered list of the $\log(\e U)$ least significant bits of every stream element in that bucket. Storing such an ordered list of length $\ell$ takes $O(1 + \ell \log (\e U))$ space, so the total space taken is at most $O(1/\e + t \log(\e U)) \le O(\inv \e \log(\e U))$, which is at most a constant multiple of the desired space.

This completes the discussion of the space taken by the sketch itself. Now we will show that the algorithm does not require any extra space (asymptotically) during the merge operation.

\paragraph{Space during merge.} During the merge, the only extra memory we require is that of storing the keys (i.e., vertices) of the map $W_{i-1}$ which weren't already stored in $F_{i-1}$. There are two parts of this: we need to store the new keys of $W_{i-1}$ (that is, the vertices with newly added weight), and we need to store the weights themselves.

Let $S$ denote the set of new keys of $W_{i-1}$. Note that every node in $S$ corresponds to at least one node from $F_i$ which put its weight into that node. Thus, we have $|S| \le |F_i|$. Additionally, $S \cup F_{i - 1}$ form an upward-closed set in $T_{i-1}$. Thus, just as we stored $F_{i-1}$, we can also store $S \cup F_{i - 1}$ using $|S \cup F_{i - 1}| \le |F_i| + |F_{i-1}|$ space. Note that we already used $|F_i| + |F_{i-1}|$ space for the original sketch, so storing $S$ does not require any more space asymptotically.

Now, it remains to store the weights in $W_{i-1}$. Here we must distinguish between the cases $i=k$ and $i < k$. If $i = k$, then we store the weights explicitly. The weights always remain at most $\al_{k-1} = O(\al_k) = O(n_0/n^*)$ (by \cref{fact:al-k-1}), so the total space required to store the weights is $O(|S| \log(n_0/n^*))$. Since $|S| \le F_k = O(1/\e)$, this is then at most the weight allocated to store $W_k$ originally, so again this does not require extra asymptotic space.

If $i < k$, then we first make one small optimization: as stated in a footnote, in \cref{alg:push-up} (the compression algorithm), we do not need to move the weight up in increments of 1. Indeed, the weights start out as multiples of $\al_i$, and the threshold $\al_{i-1}$ is also a multiple of $\al_i$. Thus, we can move weight in increments of $\al_{i-1}$, so that the weights in $W_k$ always remain multiples of $\al_i$. Now, since the weights are all multiples of $\al_i$, we can store their ratios with $\al_i$; we store the ratios in \textit{unary}, so that storing a weight of $\ell \al_i$ requires $O(\ell + 1)$ bits of space. Then, the total space needed to store the weights is $O(n_i/\al_i + |S|)$. Again, $|S| \le F_i$, so we can see that this is again at most the weight allocated to storing $F_i$ originally. 

Thus, we have shown that in all cases, the merge step does not require any more space (asymptotically) than storing the sketch already does.

\subsection{Runtime} \label{sec:runtime}

In this section, we prove that, for reasonably-sized $n$, our algorithm processes updates and queries in $O(\log (1 / \epsilon))$ amortized time. We will need a few technical assumptions and simplifications to make our algorithm run in $O(\log (1 / \epsilon))$ time. The first is that we relax the space requirement a bit to $O(\epsilon^{-1} (\log (\epsilon n) + \log U))$ bits, which still within $O(\epsilon^{-1})$ words. Secondly, we assume that $n > (\log U)^C / \epsilon^2$ for some absolute constant $C$ that depends on the computational model. Also, we assume that there are no queries during the first $(\log U)^C / \epsilon^2$ insertions.

\paragraph*{Insertion into the last layer.} Our procedure for insertion, \Cref{alg:insertion}, contains two steps. The first step is to insert the new element $x$ into the last-layer sketch $T_k$. The second step is to merge the layer $i$ into $i - 1$ (\Cref{alg:merge}).

Now, let us focus on the time complexity of the first step (\Cref{Line:insert-find,Line:insert-add} of \Cref{alg:insertion}). The reason we relax the space requirement a little is to allow us to store the tree $T_k$ at the last layer explicitly, not in the bracket representation. There are at most $3|F_k| \leq 3 \cdot \frac{n_k}{\alpha_k} = O(1 / \epsilon)$ nodes in the last layer. For each node $u \in T_k$, we store its weight $W_k[u]$ (which takes $O(\log(\epsilon n))$ bits) and the interval $[a_u, b_u]$ (which takes $O(\log U)$ bits).

To efficiently find the highest non-full node containing $x$, we always maintain a sorted list of all exposed nodes (non-full nodes whose parent is full and the non-full roots). By \Cref{obs:interval-cover}, these nodes have disjoint intervals whose union covers the entire $[U]$. Thus these nodes are simply sorted in the increasing order of these intervals. A binary search in $O(\log 1/\epsilon)$ time finds the exposed node (which is also the highest non-full node) $u$ whose interval contains $x$. Then, we increase the weight $W_k[u]$ of that node by $1$. 

In the rare case where the node $u$ becomes full after this, we need to remove it from the list and add its two empty children. Although this takes $O(1 / \epsilon)$ time as we have to modify the entire list and the topology of the tree we store, it only happens once every $\alpha_k = \frac{n_0}{n^*}$ (\Cref{eq:def-n-star}) insertions. Here $n_0$ is the current estimate of string length, which keeps doubling as explained \Cref{sec:unknown-n}. Since we know that $n > (\log U)^C / \epsilon^2$ from our assumption, we can run the algorithm starting with~$n_0 =(\log U)^C / \epsilon^2$. As $n^* \leq \epsilon^{-1} (\log(\epsilon U))^{1 + o(1)}$ (\Cref{fact:n-star-bd}), we have $\alpha_k \geq O((\log U)^{C - 1} / \epsilon)$. We can amortize the $O(1 / \epsilon)$ running time to these $\alpha_k$ insertions and get $O(1)$ amortized running time for updating the list. 

\paragraph*{Merging layer $i$ into layer $i - 1$.} First of all, in each tree $T_i$, the number of all nodes is $|F_i| \leq \frac{n_i}{\alpha_i} = O(\log(\epsilon_i U_i) / \epsilon_i)$ ($\epsilon_i = \epsilon / 2^{k - i + 4}$). We want to amortize the time cost to $n_i$ insertions. For \Cref{alg:merge}, there are three procedures which we will analyze one by one. 

\begin{itemize}
    \item $\Call{Move}{i}$ (\Cref{alg:push-down}): At \Cref{line:v_prime}, we need to find the base-level descendant $v'$ of $v$ for every node $v \in F_i$ above the base level. This can be done by traversing the stored part of tree $T_i$ once, which takes $|F_i|$ time.
    
    In the rest of this algorithm, since we only maintain the full nodes $F_{i - 1}$ in $T_{i - 1}$, in this step, all the empty nodes in $T_{i - 1}$ whose weights increase are not stored before by our algorithm. We simply store them and their weights as a list using $O((\log U + \log(\epsilon n)) \cdot |F_i|)$ bits of memory in the depth-first-search order. This takes $O(|F_i|)$ time. 
    
    \item $\Call{Compress}{i - 1}$ (\Cref{alg:push-up}): In the time efficient implementation of $\Call{Compress}{i - 1}$, instead of moving weights one unit at a time, we process the nodes in the list we stored during \Cref{alg:push-down} in top-down order and always move the maximum amount of weight that we can move. Since this process moves weight up  at most once from each node, it also only takes $O(|F_i|)$ time.
    \item $\Call{Round}{i - 1}$ (\Cref{alg:round}): Finally, \Cref{alg:round} finds the partial nodes in our list while visiting each node at most once. So this takes only $|F_i|$ time as well. 
\end{itemize}

After these three steps, we also have to update the topology of $F_{i - 1}$ and add new full nodes to its bracket representation. This takse $|F_{i - 1}|$ time. In total, the time complexity is $|F_{i - 1}| + |F_i|$. So the amortized time is $(|F_{i - 1}| + |F_i|) / n_i = O\left(1 / \alpha_i\right) \leq O(1 / \alpha_k)$ per layer $i$. As there are $k = \log^* U$ many layers, while $\alpha_k \geq (\log U)^{C - 1} / \epsilon$, the amortized time cost is just $O(1)$. 

\paragraph*{Answering rank queries.} For answering rank queries, running exactly \Cref{alg:query} requires traversing $T_0, T_1, \dots, T_k$, which takes $O(\sum_{i=0}^k |F_i|) = O((\log U) / \epsilon)$ time. For simplicity, we assume that there are only queries after first $n_0$ elements are inserted. After every $\epsilon \cdot n_0$ insertions, we run \Cref{alg:query}, compute each $\epsilon$-approximate quantile and store them. This takes at most $O((\log U) / \epsilon^2)$ time. Then for every query $x$, we just binary search in $O(\log 1 / \epsilon)$ time, and count the number of stored quantile elements less than that $x$, multiply that by $\epsilon t$ (where $t$ is the number of current insertions), and output the answer. This has an error of at most $2 \epsilon n$. Since we can amortize the $O((\log U) / \epsilon^2)$ time cost to $\epsilon \cdot n_0 \geq (\log U)^C / \epsilon^2$ elements. This takes $O(\log 1 / \epsilon)$ amortized time per query and $O(1)$ amortized time per insertion.

\section{Practical considerations} \label{sec:practice}

\paragraph{Mergeability.} One popular feature with quantile sketches is being \emph{fully-mergeable}, meaning that any two sketches with the error parameter $\epsilon$ can be merged into a single sketch without increasing the error parameter $\epsilon$. A weaker notion of mergeability is the \emph{one-way mergeability}, which, informally speaking, means that it is possible to maintain an accumulated sketch $S$ and keep merging other small sketches into $S$ without increasing the error $\epsilon$. As pointed out in \cite{greenwald2016quantiles, agarwal2013mergeable}, every quantile sketch is one-way mergeable.

Among these sketches, the GK sketch and the optimal KLL sketch is not fully mergeable, while q-digest is fully mergeable, and KLL sketch has a mode in which it is fully mergeable but loses its optimal space bound. Our sketch is based on the fully mergeable Q-digest sketch, but we do not know whether it is fully mergeable in its current form. We leave it as a future direction to come up with a fully mergeable mode for our algorithm. 

However, our algorithm is in a sense partially mergeable. That is, if we have two instances of size at most $n$ each with error parameter $\e$, we can merge them while incurring an additional discrepancy of at most $O(\e n / \log(\e U))$ (as we will soon describe). Though this is not as strong as a fully-mergeable data structure, which incurs additional error of $0$, it is still better than the $O(\e n)$ additional error incurred by merging quantile sketches in a black-box sense (by querying their quantiles to obtain an $O(\e n)$-approximation to their streams). In practice, this means that one can merge up to $\poly(U)$ of our sketches simultaneously (by performing merges in a binary tree with depth $O(\log(\e U))$), with only a constant-factor loss in $\e$.

We now sketch how to perform this partial merge. Suppose we wish to merge the data structures $D$ and $D'$, with current sizes $t > t'$. To begin with, let us first imagine that only layer 0 is occupied (in both structures). Then, we simply add values of the weight map $W'_0$ (of $D'$) into $W_0$ (of $D$). Then, the discrepancy of $W_0$ is now $\e t + \e t'$. Now, the only problem is that the invariant that all nodes are either full or empty may not hold anymore, and the full nodes are no longer upward-closed. To fix this, we perform the compression and rounding steps of \cref{alg:push-up,alg:round} --- by \cref{lemma:push-up,lemma:round}, this increases the discrepancy by at most $\al_0 = O(\e t / \log(\e U))$. If there is now a doubling step (\cref{alg:double}) to be performed (that is, if $t_0 + t_0' \ge n_0$), then we now do it as usual. Note that though the discrepancy has increased, the data structure is otherwise still a valid data structure for the error parameter $\e$, and we can continue to perform the usual operations (including more merges) on the new data structure, while keeping track of the increased discrepancy.

Now, suppose that there are occupied layers other than layer 0. Then, before merging the two data structures, we simply perform the operation \Call{Merge}{$i$} early for $i=k, k-1, \dots, 1$, on both data structures. This proceeds identically to an ordinary \Call{Merge} operation, except that during the rounding step, the total weight may not be a multiple of $\al_{i-1}$; we simply discard the excess weight down to a multiple of $\al_{i-1}$ (and insert arbitrary elements to replace them at the end of the merge). Overall, this has the effect of discarding elements down to the nearest multiple of $\al_0$, so it will introduce a discrepancy of at most $\al_0 = O(\e t / \log(\e U))$. Additionally, the proof of \cref{lemma:merge-error} still shows that the discrepancy introduced by this merge is at most $\gamma_1 n_1 = O(\e t / \log(\e U))$. Thus, overall, this partial merge still adds an additional $O(\e t / \log(\e U))$ to the discrepancy, as desired. 

\paragraph{Constant factors.} The parameters that we selected in \cref{sec:params} were chosen to make the analysis simple. There is, however, a lot of leeway in choosing the parameters to still satisfy the necessary properties, and our exact choices likely do not attain the best constant factors on space complexity. We use $k+1 = \log^*(\eps U) + 1$ layers, but in practice, we expect that around 4 layers is probably enough, and the parameters can then be chosen appropriately.

Additionally, beyond just the setting of our parameters, our analysis has generally been wasteful in terms of constants for ease of presentation and readability. There are several places this can be improved. For example, we can improve the error $\eps$ by a factor of 2 by performing the moving and rounding steps of the merge in different directions; that is, in the moving step, we can move nodes only to their leftmost (least) descendant, and in the rounding step, we round nodes upward only (which is what we already do).

\paragraph{Removing amortization.} Currently, our runtime analysis is amortized, since a step containing a merge can take a long time compared to a normal insertion step. If one is concerned about worst-case update time, then we can improve performance by executing the time-consuming operations over a longer time period while storing received elements in a buffer, similarly to Claim 3.13 of \cite{assadi2023generalizing}.

\paragraph{Answering select queries with real elements.} One feature of quantile queries is that they can also answer \textit{select} queries: that is, given a rank $r$, one can query $\select(r)$ to obtain an element $x$ that is between the rank-$(r-\e t)$ and rank-$(r+\e t)$ elements of the stream. This is equivalent to being able to answer rank queries, since one can use a binary search of rank queries to answer a select query (and vice versa). One might also desire, though, that the answers to the select queries are actual elements of the stream, rather than arbitrary elements of $[1, U]$. As stated, our algorithm does not provide a way to do this. It turns out, however, that given any quantile sketch algorithm that can answer approximate rank queries, it is possible to augment it (in a black-box manner) so that it can answer select queries with real elements of the stream, with only a constant-factor degradation in the error parameter $\e$. We will now sketch how to do so.

We initialize a quantile sketch with error parameter $\e$, and we maintain a list $x_1 < x_2 < \dots < x_\ell$ which are actual elements of the stream (and by convention we write $x_0 = 0$ and $x_{\ell+1} = U+1$), and rank estimates $r_1, \dots, r_\ell$ (where again by convention we say $r_0 = 0$) satisfying the following properties at all times $t$:
\begin{enumerate}
    \item For all $0 \le i \le \ell$, $\abs{\rank_\pi(x_i) - r_i} \le \e t$. \label{item:p1}
    \item For all $0 \le i \le \ell$, $\rank_\pi(x_{i+1}-1) - r_i \le 2\e t$. \label{item:p2}
\end{enumerate}
(Note that the first item is trivially satisfied for $i=0$.) Now, suppose that we receive an insertion $x$ into the stream. First, we increment $r_i$ for all $i$ such that $x_i \ge r_i$, to maintain property \ref{item:p1} (note that $t$ increases by 1, but this only makes property \ref{item:p1} easier to satisfy). 

Now, if $x = x_j$ for some $j$, then property \ref{item:p2} continues to be satisfied since the left-hand side of the inequality remains the same for all $i$. Otherwise, suppose that $x \in (x_j, x_{j+1})$ for some $j$. Then, \ref{item:p2} might become violated for $i=j$, since the left-hand side will have increased by 1. To fix this, we insert a new element $x_{j+1} = x$ (and shift the indices of the existing $x_i, r_i$ of all $i \ge j+1$ up by 1). Then, we execute a rank query on $x$ to get $r$ such that $|\rank_\pi(x) - r| \le \e t$. Then, we set $r_{j+1} = \max \{r, r_j + 1\}$. Note that property \ref{item:p1} continues to be satisfied by the accuracy of the rank query and because $r_j + 1 \le \rank_\pi(r_j) + \e t + 1 \le \rank_\pi(r_{j+1}) + \e t$. It remains to check that property \ref{item:p2} is now satisfied. Indeed, for $i=j+1$, this follows from the fact that $r_{j+1} \ge r_j+1$ and that the property was previously satisfied for $i=j$. For $i=j$, it follows from the fact that $\rank_\pi(x-1)$ is at most the former value of $\rank_\pi(x_{j+1}-1)$, and that the property was previously satisfied for $i=j$. Thus, we have established that the properties both continue to hold.

Finally, while there is any $j$ such that $r_{j+1} - r_{j-1} \le \e t$, we delete $x_j$ and $r_j$ (and shift the indices $i>j$ down by 1 to accommodate). This preserves the properties: we only need to check property \ref{item:p2} for $i=j-1$, and indeed, $\rank_\pi(x_j-1) - r_{j-1} \le (r_j + \e t) - r_{j-1} \le 2\e t$ by property 2 and by the assumption that $r_j - r_{j-1} \le \e t$ (note that the old $r_{j+1}$ has become $r_j$). Thus, this preserves the properties.

Now, we answer a select query as follows: on a query of rank $r$, we pick the minimal $i$ such that $r \le r_i + 2\e t$, and return $x_i$. As a special case, if $r < 2\e t$, we return $x_1$ instead of $x_0=0$. (Note that by property \ref{item:p2} applied to $i=\ell$, we never return $x_{\ell+1}$.) Then, assuming that $r \ge 2\e t$, we have by property \ref{item:p1} that $\rank(x_i) \ge r_i - \e \ge r - 3\e t$. Also, by property \ref{item:p2}, $\rank(x_i - 1) \le r_{i-1} + 2\e t < r$ (by minimality of $i$), so the rank-$r$ element is at least $x_i$. Thus the error in the select query is at most $O(\e t)$ as long as $r \ge 2\e t$. Also, in the special case $r < 2\e t$, we answer $x_1$, and by property \ref{item:p2}, $\rank(x_1 - 1) \le 2\e t$, so again the error is at most $O(\e t)$. Thus, the answers to the select queries are always approximately correct.

Finally, it remains to analyze the total space taken. Note that we have $r_{j+1} - r_{j-1} \le \e t$ for all $j$, so the total number of indices $\ell$ is at most $O(1/\e)$. Therefore, we only need to store the $O(1/\e)$ elements $x_1, \dots, x_\ell$ and $r_1, \dots, r_\ell$, which takes $O(1/\e)$ words. Indeed, since the $x_i$ are in increasing order and the increments of the $r_i$ are at most $O(\e n)$, we can actually store these in $O(\e^{-1}(\log(\e U) + \log(\e n))$ space, so this does not take any additional asymptotic space over our algorithm.

\section{Lower bounds} \label{sec:lb}
The space complexity of our algorithm is $O(\eps^{-1}(\log (\eps U) + \log (\eps n)$. In this section, we'll discuss the optimality of this result. The first term $O(\eps^{-1}\log (\eps U))$ must be incurred by any quantile sketch, even a randomized one that succeeds with reasonable probability, as we will now show. This already implies that when $n\leq \poly(U)$, our algorithm is tight\footnotemark. When this is not the case, we conjecture that our algorithm is optimal among deterministic sketches anyway. In particular, Conjecture~\ref{conj:parallel-count} implies a space lower bound of $O(\eps^{-1}\log (\eps n))$ for quantiles.

\footnotetext{Technically speaking, this result alone only implies tightness when $n \leq \poly(\e U)$. However, if $U > 1/\e^2$, then $\poly(\e U)$ and $\poly(U)$ are the same, and when $U < 1/\e^2$, then $n\leq \poly(U)$ implies that $n \ll \poly(1/\e)$, and as we discussed in \cref{sec:further-dirs}, a result of \cite{aden2022amortized} implies that our algorithm is tight when $\inv \e > \log(\e n)$.}

\begin{thm}
    Any randomized streaming algorithm for Problem~\ref{prob:all-quantile-sketch} that succeeds with probability at least $0.9$ (that is, it can answer a rank query chosen by an oblivious adversary with that probability) on a universe of size $U>C\eps^{-1}$ for some sufficiently large $C$ uses at least $\Omega(\eps^{-1}\log (\eps U))$ bits of space. 
\end{thm}

\begin{proof}
    It suffices to show that the final state of the algorithm requires $\Omega(\eps^{-1}\log (\eps U))$ bits of space. Let us restrict ourselves to streams that only contain $k=3\eps^{-1}$ distinct elements, each of which occurs $n/k$ times. Under this model, let the stream be $\pi'_1< \ldots <\pi'_k$ (each with multiplicity $n/k$). Under this model, the min-entropy of the stream (when the stream is chosen uniformly randomly) is $\log\binom{U}{k}$.
    We will show that access to the sketch reduces the min-entropy considerably (by at least a constant factor). To do this, we will describe an algorithm for a party to make $\eps^{-1}\log U$ queries to the sketch and with probability at least $0.01$, output at least $0.01$ fraction of the elements $\pi'_1, \pi'_2, \ldots, \pi'_k$ correctly. The min-entropy of this distribution of outputs is much lower: the only possibilities are those that overlap on at least $0.01$-fraction of $\pi'_1\ldots \pi'_k$, of which there are at most $\binom{k}{0.01k} \binom{U}{0.99k}$. The most likely outcome therefore occurs with probability at least $0.01$ times the $\log$ of this quantity, so the min-entropy has decreased by at least
    \[
        \log\binom{U}{k} - \log \left( 100\binom{k}{0.01k} \binom{U}{0.99k} \right) \geq \left(\Omega(\eps^{-1}\log(\eps U))\right)
    \]
    by Stirling's approximation when $U>C\eps^{-1}$ for a sufficiently large $C$. Then, by the fact blow, the sketch must have contained at least this many bits of information.

    \begin{fact}
        Let $H_{\min}(\cdot)$ denote the min-entropy of random variables. For any two random variables, $\mathbf{x}$ and $\mathbf{y}$ supported on $X$ and $Y$ respectively, we have 
        $$H_{\min}(\mathbf{x}) - H_{\min}(\mathbf{x} \mid \mathbf{y}) \leq H(\textbf{y}).$$
        In our case, $\mathbf{x}$ is the elements $\pi'_1, \pi'_2, \dots, \pi'_k$ and $\mathbf{y}$ is the memory state of our algorithm.
    \end{fact}
    \begin{proof}
    \begin{align*} 
    H_{\min}(\mathbf{x}) - H_{\min}(\mathbf{x} \mid \mathbf{y})
    &= H_{\min}(\mathbf{x}) - \sum_{y \in Y} \Pr(\mathbf{y} = y) \min_{x \in X} \log \frac{1}{\Pr(\mathbf{x} = x \mid \mathbf{y} = y)} \\
    &= H_{\min}(\mathbf{x}) - \sum_{y \in Y} \Pr(\mathbf{y} = y) \min_{x \in X} \log \frac{\Pr(\mathbf{y} = y)}{\Pr(\mathbf{x} = x, \mathbf{y} = y)} \\
    &\leq H_{\min}(\mathbf{x}) -  \sum_{y \in Y} \Pr(\mathbf{y} = y) \min_{x \in X} \log \frac{\Pr(\mathbf{y} = y)}{\Pr(\mathbf{x} = x)} \\
    &= H_{\min}(\mathbf{x})  - \min_{x \in X} \log \frac{1}{\Pr(\mathbf{x} = x)}  +  \sum_{y \in Y} \Pr(\mathbf{y} = y)  \frac{1}{\Pr(\mathbf{y} = y)}\\
    &= H(\mathbf{y})
    \end{align*}
    \end{proof}

    Now we describe the list of queries to ask the sketch to output least $0.01$ fraction of the elements $\pi'_1\ldots \pi'_k$ correctly with probability $0.01$. For each rank $i\in [k]$, binary search for the rank $i$'th element in a noise resilient way~\cite{pelc2002searching} (resilient to $0.2$ fraction of adversarial error). At the end, this must find the element at rank $i$ exactly, since each element's multiplicity is more than the permissible error. The noisy binary search must succeed whenever the fraction of error is at most $0.2$, which is true on at least $0.01$ fraction of the elements at least $0.01$ fraction of the time.
\end{proof}

\begin{thm}
    Conjecture~\ref{conj:parallel-count} implies that any deterministic streaming algorithm for Problem~\ref{prob:all-quantile-sketch} uses at least $\Omega(\eps^{-1}\log (\eps n))$ bits of space.
\end{thm}

\begin{proof}
    We will show the following. Any data structure that can compute a quantile sketch for $0.1\eps^{-1}$ on $n$ elements in the range $[\eps^{-1}]$ can also return counts of each element that are accurate to within $\pm\eps n$. Then, if there is a quantile sketch using $o(\eps^{-1}\log n)$ bits of memory, there is also a deterministic parallel approximate counter using that much space.

    Let us try to comp estimate the count of $i \in [\eps^{-1}]$. The true count of $i$ is the difference of the true ranks $r_i-r_{i-1}$, since the rank $r_j$ is the number of elements at most $j$. We query the rank of $i$ in the quantile sketch and get the answer $\widehat{r}_i$ and the rank of $i-1$ and get $\widehat{r}_{i-1}$. Then, 
    \[
        \Big| (r_i-r_{i-1}) - (\widehat{r}_i-\widehat{r}_{i-1}) \Big| \leq 0.2\eps n ,
    \]
    so we have a sufficiently accurate estimate of the count.
\end{proof}

\section*{Acknowledgments}
We would like to thank Jelani Nelson for his excellent mentorship, and specifically, for pointing us to this problem, helpful discussions, and suggestions for the manuscript. We would also like to thank the others who have provided feedback for drafts of the manuscript, including Lijie Chen, Yang Liu, and Naren Manoj. Lastly, we would like to thank Angelos Pelecanos for the scorpion lollipop \cite{scorpion}.

\bibliography{main}
\bibliographystyle{alpha}

\appendix
\section{Proof of \texorpdfstring{\cref{fact:parameter-facts}}{Fact \ref{fact:parameter-facts}}} \label{sec:parameter-proof}
Here, we will prove the various parts of \cref{fact:parameter-facts}, by showing a series of claims. Note that \cref{fact:gamma-bd} follows directly from the definitions of $\e_i$ and $\gamma_i$.

\begin{claim}[\cref{fact:pow-2}] \label{claim:pow-2}
For all $i$, $n_i$, $U_i$, $\e_i$, and $\al_i$ are powers of 2.
\end{claim}
\begin{proof}
This follows directly (inductively) from the definitions.
\end{proof}

\begin{claim}[\cref{fact:U-bd}] \label{claim:U-bd}
For all $i$, $\e_i U_i \ge 2$.
\end{claim}
\begin{proof}
For $i=0$, this follows from the assumption (made at the start of \cref{sec:main}) that $\e U$ is sufficiently large. For $i = 1$, we have $U_1 = 2\cceil{\log(\e U/8) + 1}/\e$ and $\e_1 = \e/2^{k+3}$, so $\e_1 U_1 = \cceil{\log(\e U/8) + 1} / 2^{\log^*(\e U) + 2}$, which is at least 2 again by the assumption that $\e U$ is sufficiently large. Finally, for $i \ge 2$, this follows by induction using the recursive definition of $U_i$ and the fact that $\e_{i-1} < \e_i$.
\end{proof}

\begin{claim} \label{claim:U-rel}
For all $i < k$, we have $\eps_{i+1}U_{i+1} \le 8 \log(\eps_i U_i)$.
\end{claim}
\begin{proof}
Since $\e_{i+1} \le 2\e_i$, we have by the inductive definition of $U_i$, \eqref{eq:def-U}, that
\[\eps_{i+1}U_{i+1} \le 4 \cceil{\log(\e_i U_i) + 1} \le 8 \log (\e_i U_i).\]
(Here we have used the fact that $\log(\e_i U_i)$ is a positive integer, which follows from \cref{claim:pow-2} and \cref{claim:U-bd}.)
\end{proof}

Now, define $Q_i = \e_i U_i / 16$. Then we have the following.
\begin{claim} \label{claim:Q-rel}
For all $i < k$, we have $Q_{i+1} \le \max\{\log Q_i, 8\}$.
\end{claim}
\begin{proof}
By \cref{claim:U-rel}, we have
\[Q_{i+1} = \f{\eps_{i+1}U_{i+1}}{16} \le \f{\log(\eps_i U_i)}{2} \le \f{\log(16Q_i)}{2} = \f{4 + \log Q_i}{2} \le \max\{\log Q_i, 8\},\]
as desired.
\end{proof}

\begin{claim}[\cref{fact:bottom-layer}]
$U_k = O(1/\e)$.
\end{claim}
\begin{proof}
We have $k = \log^*(\e U) \ge \log^*(Q_0)$, so if we iteratively take the logarithm of $Q_0$, we get down below 1 in at most $k$ steps. Thus, by \cref{claim:Q-rel}, we have $Q_k \le 8$, so $U_k = 16 Q_k / \e_k = O(1/\e)$.
\end{proof}

\begin{claim}[\cref{fact:tot-U}]
$U_1 + U_2 + \dots + U_k = O(\inv \e \log(\e U))$.
\end{claim}
\begin{proof}
We have $U_1 = 2\cceil{\log(\e_0 U_0 + 1)}/\e_0 = O(\inv \e \log(\e U))$. Meanwhile, for $i > 1$, by \cref{claim:Q-rel}, we have $Q_i \le O(\log \log Q_0) = O(\log \log(\e U))$. Also, $\e_i \ge 2^{-k+3}\e = \Omega(2^{-\log^*(\e U)}\e)$. Therefore, for $i > 1$, we have $U_i = O(Q_i/\e_i) \le O(\inv \e 2^{\log^*(\e U)} \log \log(\e U))$. Thus, since $k = \log^*(\e U)$,
\[U_2 + \dots + U_k \le O(\inv \e \log^*(\e U) 2^{\log^*(\e U)} \log \log(\e U)) < O(\inv \e \log(\e U)),\]
so we are done.
\end{proof}

\begin{claim}[\cref{fact:n-dec}] \label{claim:n-dec}
For all $i < k$, $n_{i+1}$ is a factor of $n_i$.
\end{claim}
\begin{proof}
Since the $n_i$ are powers of 2, it is enough to check that $n_{i+1} \le n_i$. For $i \ge 1$, this follows directly from the definition of $n_{i+1}$ since $\e_{i+1} > \e_i$ (and because of \cref{claim:U-bd}). For $i = 0$, we get $n_0 = n$ and
\[n_1 = \f{\e_0 n_0}{\e_1 \cceil{\log(\e_0 U_0) + 1}} = \f{2^{\log^*(\e U)} n_0}{\cceil{\log(\e U/8) + 1}},\]
which is at most $n_0$ by the assumption that $\e U$ is sufficiently large.
\end{proof}

\begin{claim}[\cref{fact:al-rec}] \label{claim:al-rec}
For all $i < k$, $\al_{i+1} = \al_i/\cceil{h_{i+1}+1}$.
\end{claim}
\begin{proof}
We have, by the inductive definitions \eqref{eq:def-al} and \eqref{eq:def-n}, that
\[\al_{i+1} = \f{\e_{i+1} n_{i+1}}{\cceil{h_i + 1}} = \f{\al_i}{\cceil{\log(\e_{i+1} U_{i+1}) + 1}}.\]
\end{proof}

\begin{claim}[\cref{fact:int-caps}]
If $n_0 \ge n^*$, then $\al_i, n_i \ge 1$ for all $i$.
\end{claim}
\begin{proof}
Suppose that $n_0 \ge n^*$. Firstly, by definition of $n^*$, \eqref{eq:def-n-star}, we have $\al_k \ge 1$. Also, by the definition of $\al_i$, \eqref{eq:def-al}, we also have $\al_k \le n_k$, so $n_k \ge 1$. By \cref{claim:n-dec,claim:al-rec}, $n_i$ and $\al_i$ are decreasing in $i$, so the conclusion follows.
\end{proof}

\begin{claim}[\cref{fact:n-star-bd}]
$\inv \e \le n^* \le \e^{-1} (\log(\e U))^{1+o(1)}$ (where $o(1)$ refers to a term that approaches 0 as $\e U \to \infty$).
\end{claim}
\begin{proof}
By successive applications of \cref{claim:al-rec} and then using the definition of $\al_0$, we have
\[\al_k = \f{\al_0}{\cceil{h_1+1} \cd \dots \cd \cceil{h_k+1}} = \f{n_0 \e_0}{\cceil{h_0+1} \cd \dots \cd \cceil{h_k+1}}.\]
Thus, we have
\[n^* = \f{n_0}{\al_k} = \f{\cceil{h_0+1} \cd \dots \cd \cceil{h_k+1}}{\e_0}.\]
Since $\e_0 = \e/8$, the first inequality of the claim follows immediately. Now, note that we have
\[\cceil{h_i+1} = O(\log(\e_i U_i)) = O(\max\{\log Q_i, 1\})\]
Now, this means that $\cceil{h_0 + 1} = O(\log(\e U))$, and for $i > 0$, by \cref{claim:Q-rel}, we have $\cceil{h_i + 1} \le O(\log \log \e U)$. Thus, since $k = \log^*(\e U)$, we have
\[n^* \le \f{O(\log(\e U)) \cdot (O(\log \log \e U))^{\log^*(\e U)}}{\e/8} = \e^{-1} (\log(\e U))^{1+o(1)},\]
as desired.
\end{proof}

\begin{claim}[\cref{fact:al-k-1}]
$\al_{k-1} = O(n_0 / n^*)$. 
\end{claim}
\begin{proof}
By \cref{claim:Q-rel}, we have $Q_k = O(1)$, so by \cref{fact:al-rec}, we have $\al_{k-1} = \al_k \cceil{\log Q_k + 1} = O(\al_k) = O(n_0/n^*)$.
\end{proof}

\end{document}